\documentclass{article}
\usepackage{graphicx}
\usepackage{amssymb}
\usepackage{multirow}
\usepackage{subfigure}
\usepackage{amsmath, thm-restate}
\usepackage{amsthm}
\usepackage{algorithmic}
\usepackage[section]{algorithm}
\usepackage{subfigure}
\usepackage{thm-restate}
\usepackage[normalem]{ulem}
\usepackage{fullpage}
\usepackage{color}
\usepackage{balance}  

\usepackage{url}

\begin{document}

	\author{
	Chao Li, Gerome Miklau\\
   	University of Massachusetts Amherst, Massachusetts, USA\\
   	\{chaoli, miklau\}@cs.umass.edu}

	
\title{Efficient Batch Query Answering Under Differential Privacy}
\maketitle{}

\

\abovedisplayskip = 3pt
\belowdisplayskip = 3pt
\subfigcapskip=-5pt


\floatname{algorithm}{Program}

\newcommand{\reals}{R}
\newcommand{\vol}{\textup{Vol}}
\newcommand{\convex}{\textup{Convex}}

\newcommand{\minerror}{\mbox{\sc MinError}}
\newcommand{\minsens}{\mbox{\sc MinSensitivity}}
\newcommand{\allrange}{\mbox{\sc AllRange}}
\newcommand{\allpred}{\mbox{\sc AllPredicate}}

\newcommand{\eqbydef}{\stackrel{\mathrm{def}}{=}}

\newcommand{\vect}[1]{\mathbf{#1}}
\newcommand{\sens}[1]{\Delta_{#1}}
\newcommand{\inv}[1]{{#1}^{-1}}
\newcommand{\ep}[1]{\inv{({#1}^t{#1})}}

\def\alg{\mathcal{K}}  
\def\LM{\mathcal{L}}	
\def\GM{\mathcal{G}}	
\def\MM{\mathcal{M}}	

\def\lbl{\mbox{LSA}}
\def\svdb{\mbox{\sc svdb}}
\def\cols{\mbox{cols}}

\def\tr{\mbox{trace}}
\def\var{\mbox{Var}}
\newcommand{\error}[2]{\mbox{\sc Error}_{#1}( #2 )}
\newcommand{\totalerror}[2]{\mbox{\sc TotalError}_{#1}( #2 )}
\newcommand{\maxerror}[2]{\mbox{\sc MaxError}_{#1}( #2 )}

\def\aa{\mathbb{A}}  
\def\bb{\mathbb{B}}  

\def\plus{{\!+}}
\def\b{\vect{\tilde{b}}}  
\def\x{\vect{x}}  
\def\estx{\vect{\hat x}}
\def\y{\vect{y}}
\def\q{\vect{q}}  
\def\w{\vect{w}} 
\def\v{\vect{v}}  
\def\estw{\vect{\hat w}}
\def\estq{\vect{\hat q}}
\def\A{\vect{A}}
\def\B{\vect{B}}
\def\Q{\vect{Q}}
\def\W{\vect{W}}
\def\M{\vect{M}}
\def\D{\vect{D}}
\def\P{\vect{P}}
\def\p{\vect{p}}
\def\I{\vect{I}}
\def\V{\vect{V}}
\def\H{\vect{H}}
\def\G{\vect{G}}
\def\R{\vect{R}}
\def\X{\vect{X}}
\def\Wav{\vect{Y}}
\def\lambdaB{\vect{\lambda}}
\def\LambdaB{\vect{\Lambda}}

\def\PM{\P_{\M}}
\def\DM{\D_{\M}}
\def\DS{\D_s}
\def\DSinv{\DS^{-1}}


\def\WW{\W}		
\def\Wbool{\W_{01}}		
\def\Wrang{\W_{R}}		
\def\Wunit{\W_{unit}}		

\def\real{\mathbb{R}}

\def\RR{\vect{R}}

\newcommand{\ff}[1]{#1}
\newcommand{\dif}[1]{\mathbf{\delta}_{#1}}


\newcommand{\cl}[1]{[[\emph{\color{blue}CL: #1}]]}
\newcommand{\mh}[1]{}
\newcommand{\gm}[1]{[[\emph{\color{red}GM: #1}]]}
\newcommand{\eat}[1]{}
\newcommand{\cut}[1]{}

\newcommand{\set}[1]{\{#1\}}   

\newtheorem{definition}{Definition}[section]
\newtheorem{proposition}{Proposition}[section]
\newtheorem{corollary}{Corollary}[section]
\newtheorem{conjecture}{Conjecture}[section]
\newtheorem{theorem}{Theorem}[section]
\newtheorem{problem}{Problem}[section]
\newtheorem{example}{Example}[section]
\newtheorem{remark}{Remark}[section]

\def\nbrs{nbrs}
\def\<{\langle}
\def\>{\rangle}

\def\qq{\tilde{q}}
\def\qbar{\overline{q}}

\def\Q{\mathbf{Q}}
\def\QQ{\mathbf{\tilde{Q}}}
\def\QC{\mathbf{\overline{Q}}}
\def\qq{\tilde{q}}
\def\qbar{\overline{q}}

\def\H{\mathbf{H}}
\def\HH{\mathbf{\tilde{H}}}
\def\HC{\mathbf{\overline{H}}}
\def\hh{\tilde{h}}
\def\hbar{\overline{h}}

\def\Lap{\mbox{Laplace}}
\def\Nor{\mbox{Normal}}
\def\cnt{c}
\def\cons{\gamma}
\def\db{I}

\def\hght{{\ell}}  
\def\hv{\hght(v)}
\def\wt{\alpha}
\def\root{r}

\newcommand{\frob}[1]{||#1||_f}
\newcommand{\E}{\mathbb{E}}
\newcommand{\Ldist}[3]{||#1 -#2||_{#3}}
\newcommand{\rank}{\textup{rank}}
\newcommand{\trace}{\textup{Trace}}

\newcommand{\Ltwo}[1]{||#1||_2}
\newcommand{\Lone}[1]{\left\Vert #1  \right\Vert_1}

\newcommand{\reffull}[1]{#1}

\newtheorem{lemma}{Lemma}

\newcommand{\one}[1]{\mathbb{I}_{#1}}
\def\U{\mathcal U}
\def\Z{succZ} 
\def\s{s}
\def\m{M}
\def\mm{\tilde{M}}

\vspace{-8ex}
\begin{abstract}
Differential privacy is a rigorous privacy condition achieved by randomizing query answers.  This paper develops efficient algorithms for answering multiple queries under differential privacy with low error.  We pursue this goal by advancing a recent approach called the {\em matrix mechanism}, which generalizes standard differentially private mechanisms.  This new mechanism works by first answering a different set of queries (a strategy) and then inferring the answers to the desired workload of queries.  Although a few strategies are known to work well on specific workloads, finding the strategy which minimizes error on an arbitrary workload is intractable.  We prove a new lower bound on the optimal error of this mechanism, and we propose an efficient algorithm that approaches this bound for a wide range of workloads.
\end{abstract}

\section{Introduction} \label{sec:intro}

Differential privacy \cite{Dwork:2006Calibrating-Noise} is a rigorous privacy condition, guaranteeing participants that the information released about the data will be virtually indistinguishable whether or not their personal data is included.  Differential privacy is achieved by randomizing query answers.  While there are a number of general-purpose mechanisms for satisfying differential privacy \cite{Dwork:2011A-firm-foundation}, there are comparatively few results showing that these mechanisms are optimally accurate---that is, that the least possible distortion has been introduced to satisfy the privacy criterion.  For a single numerical query, the addition of appropriately-scaled Laplace noise satisfies $\epsilon$-differential privacy and has been proven optimally accurate \cite{ghosh2009universally}.  For workloads of multiple queries, optimally accurate mechanisms are not known.  

Our focus is on batch query answering, in which multiple queries are answered at one time, in a single interaction with the private server.  A batch of queries, or a workload, here consists of a set of linear counting queries.  These include predicate counting queries, histograms, sets of marginals, data cubes, or any combination of these.  

The goal of research in this area is to devise an efficient algorithm that can achieve the least possible error under differential privacy.  In this work we pursue this goal by advancing a recently-proposed technique called the matrix mechanism \cite{Li:2010Optimizing-Linear}.  The matrix mechanism generalizes standard differentially private output perturbation techniques, and we explain it by comparing it with the standard Laplace mechanism. 

The Laplace mechanism answers a workload of queries by adding to each query a sample chosen independently at random from a Laplace distribution.  The distribution is scaled to the sensitivity of the workload (the maximum possible change to the query answers induced by the addition or removal of one tuple).  Consider using the Laplace mechanism to simultaneously release answers to the set of all range-count queries over a database containing ages for a community.  This workload consists of all queries $AgeCount(a,b)$ which return the number of individuals whose age is between $a$ and $b$, for any constants $a,b \in\{1,\dots 120\}$.

Using the Laplace mechanism directly results in extremely noisy query answers for this workload because the noise added to {\em each} query in the workload is proportional to the sensitivity of the workload.  The sensitivity of this workload is $O(n^2)$ where $n$ is the size of the domain (120 in this example).  In addition, because independent noise is added to each query, the answers are inconsistent: e.g. the sum of $AgeCount(20,40)$ and $AgeCount(41,60)$ will not, in general, equal $AgeCount(20,60)$ as one might hope. 

An alternative to the direct application of the Laplace mechanism arises from recognizing that the workload of all range-count queries can be computed from a smaller set of query answers---namely counts for each individual age $1, 2, \dots 120$.  Because each count is independent, this approach has very low sensitivity and does not introduce inconsistency.  However, the released query results must be summed to get the answers to the desired range queries: e.g. $AgeCount(1,10)$ is the sum $AgeCount(1,1) + \dots + AgeCount(10,10)$.  Although each individual count has low error, the expected error accumulates when computing the sum, leading to significant error for large ranges.

The matrix mechanism (so named because workloads are represented as matrices and analyzed algebraically) can be seen as a generalization of the multi-query Laplace mechanism.  The approaches above can be seen as two extremes: one in which the workload is submitted to the Laplace mechanism, and one in which the workload is divided into independent queries and those are submitted.  The matrix mechanism encompasses both of these extremes, along with many other approaches, some offering significantly lower error.  

Given a workload of queries, the matrix mechanism uses the Laplace mechanism to answer a different set of queries, called a {\em strategy}. The answers to the strategy queries are then used to derive the query answers ultimately desired---the workload queries.  If there are related queries in the strategy, linear regression is used to combine the evidence from all available query answers (in an optimal way) to produce the final answers.  The result is a consistent final answer to the workload queries, often with improved error.
The matrix mechanism can improve error because the strategy queries can be used to avoid or reduce redundancy that may be present in the workload, thus lowering sensitivity.  Also, redundancy that does exist in the strategy queries is exploited by the linear regression process to improve accuracy. 

Using the matrix mechanism requires instantiating it with a set of strategy queries which is a good match for the workload.  For the workload of all range queries, two approaches were independently proposed recently, one based on a wavelet transform \cite{xiao2010differential}, and one based on a hierarchical query set \cite{Hay:2010Boosting-the-Accuracy}.  While these approaches look quite different at first glance, they are in fact different sets of strategy queries which can be analyzed as instances of the matrix mechanism \cite{Li:2010Optimizing-Linear}.  Both strategies result in much lower error: $O(log^3n)$ instead of $O(n^4)$ for the workload itself or $O(n)$ for the approach that asks for individual counts.

Thus research to date has shown that for a particular workload (the set of all range-count queries) there are strategies that offer significantly lower error.   But these strategies do not work well for all workloads.  To exploit the full power of the matrix mechanism, we must customize strategies to the given workload.  

Unfortunately, we cannot hope for exact solutions to this problem.  The {\em strategy design problem}, i.e., calculating the strategy that results in the minimal error for a given workload, is intractable~\cite{Li:2010Optimizing-Linear}.  Nevertheless, in this paper we provide efficient algorithms for computing a set of strategy queries for a workload and show that they approach optimally accurate strategies.

A key aspect of our approach is the move from standard $\epsilon$-differential privacy, to $(\epsilon,\delta)$-differential privacy, a modest relaxation of the privacy condition often called {\em approximate} differential privacy.  The matrix mechanism is easily adapted to approximate differential privacy.  Computing the strategy with minimal error remains computationally infeasible.  But we show that under this definition the matrix mechanism has a number of nice features which make it amenable to analysis and which lead to better approximate solutions.

Our contributions are organized as follows.  The central theoretical result, shown in Sec. 3, is a tight lower bound on the minimum total error achievable for a workload.  It can be efficiently computed from the spectral properties of a workload when it is represented in matrix form.  In Sec. 4, we propose an efficient localized search algorithm for solving the strategy design problem.  In Sec 5, we show experimentally that, for a variety of workloads, our algorithm produces strategies that approach the optimal.

We use our bound on error to understand {\em workload} {\em complex\-ity}---that is, the relative difficulty of simultaneously answering a workload of queries accurately.  For instance, sets of multi-dimensional range queries are ``easier'' to answer than one dimensional range queries.  We also use the error bound to evaluate the quality of existing approaches.  The hierarchical and wavelet strategies mentioned above are fairly close to optimal in some cases, but can be significantly improved in others.  For example, for two-dimensional range queries, the wavelet strategy results in error $2.2$ times the optimal while our algorithm finds a strategy that is just $10\%$ greater than the optimal. 
	Ultimately, we conclude that adapting the matrix mechanism to the specific properties of a workload is crucial: there are workloads for which the total error achieved using our algorithm is an order of magnitude less than that of existing techniques.



\section{Background}

In this section we describe our data model and privacy conditions. We also review the fundamentals of the matrix mechanism, including error measurement and our main problem of strategy design. 

\subsection{Data model and linear queries}
The database $I$ is an instance of a single-relation schema $R(\mathbb{A})$, with attributes $\mathbb{A}=\{A_1, A_2, \ldots, A_m\}$.  The crossproduct of the attribute domains, written $dom(\mathbb{A})$, is the set of all possible tuples.  In order to express our queries, we first construct from $I$ a vector $\x$ of {\em cell counts}.  Each element $x_i$ of $\x$ is associated with $cell(x_i)$, a disjoint subset of $dom(\mathbb{A})$.  Then $x_i$ is the count of the tuples from $cell(x_i)$ that are present in $I$: $x_i = |\{t \in I | t \in cell(x_i)\}|$.  All queries are expressed using the cell counts in $\x$.  We always use $n$ for the size of $\x$, which we sometimes refer to simply as the domain size.

One way to define the vector $\x$ is to choose the smallest possible cells: one cell for each tuple in $dom(\mathbb{A})$.  This is often inefficient (the size of the $\x$ vector is the product of the attribute domain sizes) and ineffective (the base counts are typically too small to be estimated accurately under the privacy condition).  Instead, we often consider queries over larger cells.  A common way to form a vector of base counts is to partition each $dom(A_i)$ into $d_i$ regions, which could correspond to ranges over an ordered domain, or individual elements (or sets of elements) in a categorical domain.  Then the individual cells are defined by taking the cross-product of the regions in each domain.  Other alternatives are possible as the cells need not be formed from contiguous subsets of the attribute domains as long as all cells are disjoint.  Constructing a vector of base counts appropriate to a workload is usually straightforward but we provide a practical example in App. \ref{app:datamodel} to aid the reader.  

A linear query computes a specified linear combination of the cell counts in $\x$.
\begin{definition}[Linear query]
A {\em linear query} is a length-$n$ row vector $\q=[q_1 \dots q_n]$ with each $q_i \in \mathbb{R}$.  
The answer to a linear query $\q$ on $\x$ is the vector product $\q\x = q_1x_1 + \dots + q_nx_n$.
\end{definition}


A set of queries is represented as a matrix, each row of which is a single linear query.

\begin{definition}[Query matrix]
A {\em query matrix} is a collection of $m$ linear queries, arranged by rows to form an $m \times n$ matrix.
\end{definition}

If $\W$ is an $m \times n$ query matrix, the query answer for $\W$ is a length $m$ column vector of query results, which can be computed as the matrix product $\W \x$.

A {\em workload} is a query matrix representing a batch of linear queries of  interest to a user.  
We introduce notation for two common workloads which contain all queries of a certain type.  $\allrange(d_1, \dots d_k)$ refers to the set of all multi-dimensional range-count queries over $k$ ordered domains, where each is divided into $d_i$ regions.  Here the size of the cell count vector, $n$, is the product $\prod_{i=1}^k{d_i}$.  {\sc All\-Pred\-icate}$(n)$ is the set of all predicate counting queries over $n$ cells.  $\allpred(n)$ contains all $2^n$ linear queries of size $n$ with coefficients of 0 or 1.  We also consider workloads consisting of arbitrary subsets of each of these types of queries, low-order marginals, and their combinations.

\subsection{Privacy definitions and mechanisms}

Standard $\epsilon$-differential privacy \cite{Dwork:2006Calibrating-Noise} places a bound (controlled by $\epsilon$) on the difference in the probability of query answers for any two {\em neighboring} databases.  For database instance $\db$, we denote by $\nbrs(\db)$ the set of databases differing from $\db$ in at most one record. 
Approximate differential privacy~\cite{Dwork:2006Our-Data-Ourselves:,McSherry:2009fk}, is a modest relaxation in which the $\epsilon$ bound on query answer probabilities may be violated with small probability (controlled by $\delta$).

\begin{definition}[Approximate Differential Privacy] A randomized algorithm $\alg$ is $(\epsilon,\delta)$-differentially private if for any instance $I$, any $I' \in \nbrs(I)$, and any subset of outputs $S \subseteq Range(\alg)$, the following holds:
\[
Pr[ \alg(I) \in S] \leq \exp(\epsilon) \times Pr[ \alg(I') \in S] +\delta
\]		
	\end{definition}

Both definitions can be satisfied by adding random noise to query answers.  The magnitude of the required noise is determined by the {\em sensitivity} of a set of queries: the maximum change in a vector of query answers over any two neighboring databases.  However, the two privacy definitions differ in the measurement of sensitivity and in their noise distributions.  Standard differential privacy can be achieved by adding Laplace noise calibrated to the $L_1$ sensitivity of the queries \cite{Dwork:2006Calibrating-Noise}. Approximate differential privacy can be achieved by adding Gaussian noise calibrated to the $L_2$ sensitivity of the queries \cite{Dwork:2006Our-Data-Ourselves:,McSherry:2009fk}.  This small difference in the sensitivity metric---from $L_1$ to $L_2$---has important consequences for our algorithms and our theoretical results.  Our results focus on approximate differential privacy, but App. \ref{app:l1l2} contains a thorough comparison of these two definitions as they pertain to the matrix mechanism and the results of this paper.


Since our query workloads are represented as matrices, we express the sensitivity of a workload as a matrix norm.  Because neighboring databases $\db$ and $\db'$ differ in exactly one tuple, and because cells are disjoint, it follows that the corresponding vectors $\x$ and $\x'$ differ in exactly one component, by exactly one, in which case we write  $\x' \in \nbrs(\x)$.  The $L_2$ sensitivity of $\W$ is equal to the maximum $L_2$ norm of the columns of $\W$. Below, $\cols(\W)$ is the set of column vectors $W_i$ of $\W$.

\begin{proposition}[$L_2$ Query matrix sensitivity]
The $L_2$ sensitivity of a query matrix $\W$ is denoted $\Ltwo{\W}$, defined as follows:
\begin{eqnarray*}
\Ltwo{\W} & \eqbydef & \max_{\x' \in \nbrs(\x)} \Ltwo{\W\x - \W\x'} 
			= \max_{W_i \in \cols(\W)} \Ltwo{W_i}
\end{eqnarray*}
\end{proposition}




The classic differentially private mechanism adds independent noise calibrated to the sensitivity of a query workload.  We use $\Nor(\sigma)^m$ to denote a column vector consisting of $m$ independent samples drawn from a Gaussian distribution with mean $0$ and scale $\sigma$.

\begin{proposition}{\sc (Gaussian mechanism \cite{Dwork:2006Our-Data-Ourselves:, McSherry:2009fk})}\label{thm:l2diffpriv}
Given an $m \times n$ query matrix $\W$, the randomized algorithm $\GM$ that outputs the following vector is $(\epsilon,\delta)$-differentially private:
$$\GM(\W,\x) = \W\x + \Nor(\sigma)^m$$ 
where $\sigma=\Ltwo{\W}\sqrt{2\ln(2/\delta)}/\epsilon$
\end{proposition}

Recall that $\W\x$ is a vector of the true answers to each query in $\W$.  The algorithm above adds independent Gaussian noise (scaled by the sensitivity of $\W$, $\epsilon$, and $\delta$) to each query answer.  Thus $\GM(\W,\x)$ is a length-$m$ column vector containing a noisy answer for each linear query in $\W$.

The matrix mechanism has a similar form as the algorithm above, but adds a more complex noise vector.  It uses a different set of queries (the strategy matrix, $\A$) to construct this vector.

\begin{proposition}{\sc ($(\epsilon,\delta)$-Matrix Mechanism \cite{Li:2010Optimizing-Linear})} \label{def:m-mech} 
Given an $m \times n$ query matrix $\W$, and assuming $\A$ is a full rank $p \times n$ strategy matrix, the randomized algorithm $\MM_\A$ that outputs the following vector is $(\epsilon,\delta)$-differentially private:
\begin{eqnarray*}
\MM_\A(\W,\x) &=& \W\x + \W \A^\plus \Nor(\sigma)^m.
\end{eqnarray*}
where $\sigma=\Ltwo{\A}\sqrt{2\ln(2/\delta)}/\epsilon$
\end{proposition}

Here $\A^\plus$ is the pseudo-inverse of $\A$: $\A^\plus = \inv{(\A^T\A)}\A^T$; if $\A$ is a square, then $\A^\plus$ is just the inverse of $\A$.  The intuitive justification for this mechanism is that it is equivalent to the following three-step process: (1) the queries in the strategy are submitted to the Gaussian mechanism; (2) an estimate $\estx$ for $\x$ is derived by computing the $\estx$ that minimizes the squared sum of errors (this step consists of standard linear regression and requires that $\A$ be full rank to ensure a unique solution); (3) noisy answers to the workload queries are then computed as $\W\estx$.  The answers to $\W$ derived in step (3) are always consistent because they are computed from a single noisy version of the cell counts, $\estx$.

Like the Gaussian mechanism, the matrix mechanism computes the true answer vector $\W\x$ and adds noise to each component.  But a key difference is that the scale of the Gaussian noise is {\em calibrated to the sensitivity of the strategy matrix $\A$, not that of the workload}.  In addition, the noise added to query answers is no longer independent, because the vector of independent Gaussian samples is transformed by the matrix $\W\A^\plus$.  




\begin{example}  

The full rank strategy matrix with least sensitivity is the identity matrix, $\I$, which has sensitivity 1.  With $\A=\I$, the matrix mechanism formalizes the approach mentioned in Sec. \ref{sec:intro} which computes individual counts and sums them to answer range queries.  At the other extreme, the workload itself can be used as the strategy: $\A=\W$.  In this case, there is no benefit in sensitivity over the Gaussian mechanism.  For many workloads, neither of these naive strategies offer optimal error.  For $\W=\allrange(n)$, two strategies were recently proposed. The {\em hierarchical} strategy includes the total sum over the whole domain, the count of each half of the domain, and so on, terminating with counts of individual elements of the domain.  The {\em wavelet} strategy consists of the matrix describing the Haar wavelet transformation.  Informally, these achieve low error because they each have low sensitivity, $O(logn)$, and every range query can be expressed as a linear combination of few strategy queries.

\end{example}

\subsection{Error of the Matrix Mechanism} \label{sec:sub:error}

We measure the accuracy of differentially private query answers using mean squared error.  For a workload of queries, the error is defined as the total of individual query errors. 

\begin{definition}[Query and Workload Error] Let $\estw$ be the estimate for query $\w$ under the matrix mechanism using query strategy $\A$.  That is, $\estw=\MM_\A(\w,\x)$.  The mean squared error of the estimate for $\w$ using strategy $\A$ is: $$\error{\A}{\w} \eqbydef \E[ ( \w\x - \estw)^2 ].$$ 
Given a workload $\W$, the total mean squared error of answering $\W$ using strategy $\A$ is: $\error{\A}{\W} =$ \\
$\sum_{\w_i \in \W} \error{\A}{\w_i}$.
\end{definition}

The query answers returned by the matrix mechanism are linear combinations of noisy strategy query answers to which independent Gaussian noise has been added.  Thus, as the following proposition shows, we can directly compute the error for any linear query $\w$ or workload $\W$:


\begin{proposition}{\sc (Total Error)} \label{prop:totalerror}
Given a workload $\W$, the total error of answering $\W$ using the $(\epsilon,\delta)$ matrix mechanism with query strategy $\A$ is:
\begin{equation}\label{eqn:totalerror}
 \error\A{\W} = P(\epsilon, \delta)|| \A ||_2^2 \;\tr (\W^T\W(\A^T\A)^{-1})
 \end{equation}
where $P(\epsilon, \delta)=\frac{2\log(2/\delta)}{\epsilon^2}$.
\end{proposition}
\eat{According to proposition~\ref{prop:totalerror}, when the privacy parameters $\epsilon$ and $\delta$ are chosen, the total error of using strategy $\A$ to answer workload $\W$ is explicitly determined by 
\[|| \A ||_2^2 \tr (\W^T\W(\A^T\A)^{-1}).\]
To make the representation easier, in the rest of this paper, we always assume $P(\epsilon, \delta)=1$ so that
\[ \totalerror\A{\W} = || \A ||_2^2\tr (\W^T\W(\A^T\A)^{-1}).\]}

The $trace$ is the sum of the elements on the main diagonal of a square matrix.  Our main objective is to minimize this formula, which is determined by the relationship between $\A$ and $\W$.  Note that $\x$ (the vector of cell counts corresponding to the database) does not appear in this expression.  The minimum error strategy depends on the workload alone, not on the database instance.

\subsection{The strategy design problem}\label{sec:probstat}

The optimal strategy for a workload $\W$ is defined to be one that minimizes the total error under the $(\epsilon, \delta)$-matrix mechanism.
\begin{problem}{\sc (Minimizing Total Error)}  \label{prob:mintotal}
Given a workload $\W$, find a query strategy $\A_0$ such that:
\begin{equation}\label{eqn:mintotalerror}
	\error{\A_0}\W = {\arg\min}_\A\error\A\W.
\end{equation}
\end{problem}

The exact solution to Problem \ref{prob:mintotal} can be computed using a semi-definite program (SDP) \cite{Li:2010Optimizing-Linear}.  However, finding the solutions of the program with standard SDP solvers takes $O(n^8)$ time, making it infeasible for common applications.  Therefore, our goal in this paper is to efficiently find approximations to the optimal strategy, for any provided workload.

\section{Theoretical Analysis}\label{sec:theory}

In this section we present novel theoretical results which provide a foundation for the algorithms that come later, and which aid in evaluating both the quality of strategy matrices and the complexity of workloads.

%

\subsection{Analysis of Strategies}

An analysis of strategies allows us to characterize those which are possible solutions to the strategy design problem.

\begin{definition}[Partial Order on Strategies]\label{def:partialorder}
The following relation defines a partial order on strategies.\vspace{1ex}\\ 
\textbf{Strategy Equivalence:} Two strategy matrices $\A_1$ and $\A_2$ are equivalent ($\A_1 \equiv \A_2$) if for all linear queries $\q$, 
\[\error{\A_1}{\q}=\error{\A_2}{\q}.\]  
\textbf{Strategy Efficiency:} Strategy $\A_1$ is more efficient than $\A_2$, written $\A_1 \leq \A_2$, if for all linear queries $\q$, 
\[\error{\A_1}{\q}\leq\error{\A_2}{\q}.\]
$\A_1$ is strictly more efficient than $\A_2$ if for all queries $\q$, $\error{\A_1}{\q}$ $< \error{\A_2}{\q}$ .
\end{definition}

A strategy can never be the solution to the strategy design problem (Problem \ref {prob:mintotal}) if there is another strategy that is more efficient than it.  We call a strategy \textsl{minimal} if it is a minimal element in the partial order above.  

We show next that the set of minimal strategies has a straightforward characterization for the $(\epsilon,\delta)$ matrix mechanism.  Notice that since the sensitivity of a strategy is measured according to its maximum column norm, a strategy with some columns that do not match the maximum is wasteful: we could add queries containing non-zero entries in the deficient columns without increasing sensitivity.  These added queries will provide additional evidence that can be used to reduce overall error, so completing deficient columns can never result in a worse strategy.  

\begin{definition}[Column Uniform Matrix]
A matrix is~\textsl{column-uniform} with respect to $L_p$ if each of it columns has the same $L_p$ norm.  
\end{definition}

This suggests that column uniformity is a necessary condition for a minimal strategy.  But we can also show that it is a sufficient condition. (The proof is included in App.~\ref{app:propstrategy}.)
\begin{restatable}{theorem}{thmcolumnunif} \label{thm:columnunif}
A strategy matrix $\A$ is minimal iff it is column-uniform.
\end{restatable}

Another important property of the $(\epsilon,\delta)$ matrix mechanism is that redundant queries do {\em not} lead to less efficient strategies.  We say a query in strategy matrix $\A$ is {\em redundant} if it is a linear combination of another query in $\A$.  In this case the queries provide the same evidence about the database, but one may be scaled relative to the other, resulting in lesser or greater accuracy.  The following theorem shows that a strategy with a redundant queries is equivalent to a strategy with the redundancy removed.  

\begin{theorem}[Redundant Queries]\label{thm:redqueries}
Suppose strategy $\A_1 = \{\A_0 \cup \q \cup c_1\q\}$ for some strategy $\A_0$, some linear query $\q$, and some constant $c_1$.  Then $\A_1$ is equivalent to the reduced strategy $\A_2 = \{\A_0 \cup c_2\q\}$ where $c_2 = \sqrt{1+c_1^2}$.
\end{theorem}

The fact that the presence of redundant queries does not lead to less efficient strategies has important consequences for efficient strategy design algorithms.  Because of this property, an algorithm can make a local choice to add a query to a strategy, adding the same query again later to augment its weight if necessary.  

Thm.~\ref{thm:columnunif} and Thm.~\ref{thm:redqueries} do not hold for the $\epsilon$-matrix mechanism.  Please see App. \ref{app:l1l2} for details.

\subsection{The singular value bound}

Next we explain our main theoretical result, a lower bound on the error for a workload.  We first define the singular value decomposition of a workload.

\begin{definition}[Decomposition of Workload]
Let $\W$ be any $m \times n$ query workload.  The singular value decomposition (SVD) of $\W$ is a factorization of the form $\A = \Q_\W \D_\W \P_\W^T$ such that $\Q_\W$ is a $m \times m$ orthogonal matrix, $\D_\W$ is a $m \times n$ diagonal matrix and $\P_\W$ is a $n \times n$ orthogonal matrix.  When $m > n$, the diagonal matrix $\D_\W$ consists of an $n \times n$ diagonal submatrix combined with $\vect{0}^{(m-n) \times n}$.
\end{definition}


The bound on the total error for a workload is derived from its singular value decomposition.

\begin{restatable}{theorem}{thmsvdbound}{\sc (Singular Value Bound)}\label{thm:singularvaluebound}
Given an $m \times n$ workload $\W$, let $\lambda_1, \lambda_2,\ldots, \lambda_n$ be the singular values of $\W$.
\[\min_\A\error{\A}\W\geq P(\epsilon, \delta)\frac{1}{n}(\sum_{i=1}^n\lambda_i)^2,\]
where $P(\epsilon, \delta)=\frac{2\log(2/\delta)}{\epsilon^2}$.
\end{restatable}

A strategy matrix $\A$ can be considered to have two parts: its eigenvalues and its eigenvectors.  When $\A$ is column-uniform, the total error can be represented as a polynomial of those two parts.  The key idea behind the SVD bound is to ignore the constraint that $\A$ is column uniform and choose eigenvalues and eigenvectors separately to minimize the polynomial, which leads to an under-estimate of the total error. The details of this proof are presented in App.~\ref{app:propstrategy}. 

We use $\svdb(\W)$ as shorthand for the singular value bound of workload $\W$. Given a workload $\W$, $\svdb(\W)$ can be computed easily using standard methods for matrix decomposition, or it can be computed directly from $\W^T\W$: 
\begin{corollary}
Given an $n \times n$ positive semi-definite matrix $\W^T\W$, let $\lambda_1, \ldots, \lambda_n$ be the eigenvalues of $\W^T\W$.
\begin{equation} \label{eq:svdb2}
\min_\A\error{\A}\W\geq P(\epsilon, \delta)\frac{1}{n}(\sum_{i=1}^n\sqrt\lambda_i)^2,
\end{equation}
where $P(\epsilon, \delta)=\frac{2\log(2/\delta)}{\epsilon^2}$.
\end{corollary}

For $m \times n$ workload $\W$, computing $\svdb(\W)$ takes $O(mn^2)$ time on $\W$ itself and takes $O(n^3)$ time on $\W^T\W$, which significantly reduces the running time when the number of queries, $m$, is much larger than $n$. In the case of large regular workloads such as $\allrange$ and $\allpred$, $\W^T\W$ can be computed directly.  Then computing $\svdb(\W)$ improves from $O(n^5)$ to $O(n^3)$ for $\allrange$ and $O(n^22^n)$ to $O(n^3)$ for $\allpred$. 

\begin{example} For $\allrange(1024)$ and $\allrange(32, 32)$, the SVD bound on the total error is $5.32\times10^6$ and $4.39\times 10^6$, respectively. Below we report, as a ratio of the SVDB bound, the total error of these workloads for a number of known strategy matrices: the workload itself, the identity, hierarchical and wavelet: \\
\begin{tabular}{c|cccc}
&workload & identity & hierarchical & wavelet\\
 \hline $1024$ & $50.58$ & $33.75$ & $2.14$ & $1.84$\\
$32\times 32$ & $17.25$ & $8.15$ & $ 2.92$ &  $2.23$\\
\end{tabular}\\
In the next section we present an algorithm that finds better strategies: the ratio of total error on $\allrange(1024)$ and $\allrange(32, 32)$ using the strategies computed by the algorithm is $1.26$ and $1.08$, respectively.
\end{example}

\subsection{Properties of the singular value bound}

The SVD bound has a number of properties which make it a reliable measure of workload complexity.  Notice that in the expression for $\error\A{\W}$ in Prop. \ref{prop:totalerror}, the workload $\W$ appears only in the trace term, as $\W^T\W$.  An immediate implication is that a strategy that achieves minimal total error for a given workload $\W_1$ achieves minimal error for any workload $\W_2$ such that $\W_1^T\W_1 = \W_2^T\W_2$.  We therefore define the following notion of equivalence:

\begin{definition}[Workload Equivalence] An $n \times m_1$ workload $\W_1$ and an $n \times m_2$ workload $\W_2$ are equivalent, denoted $\W_1 \equiv \W_2$, iff $\W_1^T\W_1 = \W_2^T\W_2$.
\end{definition}

Since the SVD bound is determined by the singular values, it follows immediately that equivalent workloads have equal error bounds. That is, if $\W_1 \equiv \W_2$, then $\svdb(\W_1)=\svdb(\W_2)$.  In addition, as we would hope, the SVD bound of a workload increases monotonically under the addition of new queries.  Thus if $\W_1$ is a workload and $\W_2$ is a workload that results from adding one or more linear queries to the rows of $\W_1$, then $\svdb(\W_1) \leq \svdb(\W_2)$.  A more general result is shown below, accounting for the fact that the larger workload may be the augmentation of any workload equivalent to the smaller workload.  (The proof is omitted.)

\begin{restatable}{theorem}{thmcontainedwkld}\label{thm:containedwkld}
Let $\W_1$, $\W_2$ be workloads. If there exists a workload $\W'_2$ equivalent to $\W_2$ and the rows of $\W_2'$ contain all the rows of $\W_1$, then $\svdb(\W_1) \leq \svdb(\W_2)$.
\end{restatable}

Lastly, and most importantly, the singular value bound is tight.  For a certain set of {\em variable agnostic} workloads, it is possible to directly construct a strategy achieving the bound.  Intuitively, variable agnostic workloads treat every variable in the domain in an equivalent manner.  The workload $\allpred(n)$ is variable agnostic, but $\allrange(n)$ is not because, for example, variables in the middle of the domain occur more often in the set of all range queries.  In App. \ref{app:varagnostic} we show how to construct the optimal strategy for any variable agnostic workload, including for $\allpred(n)$.

We do not know if the singular value bound is achievable for every workload (namely those that are not variable-agnostic).  However experimentally we find that we can compute strategies that approach this bound.


\section{Strategy Selection Algorithms} \label{sec:algorithm}

In this section we present an algorithm, along with a set of performance optimizations, for computing a close-to-optimal strategy for a given workload.

\subsection{The Level Selection Algorithm (\lbl)}
The Level Selection Algorithm (\lbl) takes as input a workload and returns a strategy matrix designed to offer low error for the workload.  $\lbl$ is a localized search algorithm which builds a strategy matrix by choosing, at each step, the {\em level} of queries whose addition maximally reduces error for the workload.  A level is a set of linear queries consisting of coefficients 0 or 1, determined by a partitioning of the variables, so that each variable appears in exactly one query.  Recall that, according to Thm.~\ref{thm:columnunif}, minimal strategy matrices are precisely those that are column uniform.  By constructing a strategy by levels, the $\lbl$ algorithm always chooses among minimal strategies. 

To construct a new level, the $\lbl$ algorithm starts with the simplest level: the query $[1,1, \dots 1]$ which is the sum of all cells.  The algorithm then performs a top-down search, recursively bisecting the query into smaller queries such that the error of the workload is maximally reduced.  Once a level is completed, the $\lbl$ algorithm computes the total error of the workload with and without that level.  If the total error is not improved by the level, the algorithm terminates and outputs the current strategy.  Otherwise the level is added to the output strategy, and, subject to a user-defined threshold $k$ on the number of levels, the next level is then constructed.  

The search space of $\lbl$ is quite general.  Both the hierarchical strategies from \cite{Hay:2010Boosting-the-Accuracy} and (a strategy equivalent to) the wavelet strategy \cite{xiao2010differential} can be constructed by levels.  But the $\lbl$ algorithm is not limited to hierarchical strategies because \eat{each level is constructed independently and}there is no constraint that lower level queries are contained in higher level queries.  Further, while levels are constructed initially with coefficients of 0 or 1, the final output strategy may have a complex set of weights on queries, because redundant queries may be added in different levels.  These redundant queries can be combined, as shown in Thm.~\ref{thm:redqueries}, by appropriate weighting.


\algsetup{linenosize=\small}

\begin{algorithm}[t]\label{alg:lbl}
\small
\caption{The Level Selection Algorithm (\lbl)}
\textbf{Input:} A workload matrix $\W$ and the size of each dimension of the domain. An upper bound $k$ to the maximum number of levels.\\
\textbf{Output:} A strategy matrix $\A.$
\begin{algorithmic}[1]
\STATE $\A=\I$;
\REPEAT
	\STATE $\dif\A=[1,1,\ldots, 1]$;
	\REPEAT
		\FOR{each query $\q$ \textbf{in} $\dif\A$}
			\STATE find a pair $(i,j)$ such that split at position $i$ on dimension $j$ such that using query:
			\begin{eqnarray*}
			\q_1&=&\{\mbox{buckets with value $1$ in }\q\mbox{ with}\\
			 &&\mbox{$i$-th dimension less than $j$}\}\\
			\q_2&=&\{\mbox{buckets with value $1$ in }\q\mbox{ with}\\
			 &&\mbox{$i$-th dimension greater than or equal to $j$}\}
			\end{eqnarray*}
to take the place of query $\q$ in $\dif\A$ and\\
$\A'=\left[\begin{smallmatrix}\A\\ \dif\A\end{smallmatrix}\right]$ minimizes
$\error{\A'}{\W}$.
		\ENDFOR
		\STATE \textbf{if} $\error{\A'}{\W}$ is reduced \textbf{then} update $\dif\A$,
	\UNTIL{$\dif\A$ was not updated;}
	\STATE \textbf{if} $\error{\A'}{\W}\leq \error{\A}{\W}$ \textbf{then} $\A=\A'$.
\UNTIL{$\A$ was not updated or $||\A||_2^2 \geq k$;}
\RETURN$\A$.
\end{algorithmic}
\end{algorithm}
\begin{example}
For $\W=\allrange(512)$, the $\lbl$ algorithm terminates at $33$ levels.  The output strategy consists of $4746$ queries, reducible to $1931$ non-redundant queries. 
\end{example}
\subsection{Complexity of the $\lbl$ algorithm}
Note that Program~\ref{alg:lbl} can use $\W^T\W$ instead of $\W$ as input without impacting any computation or the final output. This implies that equivalent workloads produce equivalent outputs, and also that the complexity of Program~\ref{alg:lbl} does not depend on the number of queries in $\W$.  

In each iteration, the algorithm needs to find the split point which maximizes the reduction of total error.  This occurs in Step 6 and it is the most time-consuming part of Prog.~\ref{alg:lbl}.  The total error for each possible split has to be computed, which means computing the total error $O(n)$ times to determine one split point. There are at most $n$ split points so the total error is computed at most $O(n^2)$ over all the iterations.  Each total error computation requires recomputing $\inv{(\A'^T\A')}$ for an updated strategy $\A'$.  This requires $O(n^3)$ time by ordinary matrix inversion, so the entire algorithm would take $O(kn^5)$ time, where the number of levels is bounded by $k$.  

However, because only three rows  of $\dif\A$ are updated during each iteration, we can improve the running time by exploiting a more efficient incremental computation of the matrix inverse.  Details of this technique and the proof of the following theorem are included in App.~\ref{app:lbl}.

\begin{restatable}{theorem}{thmlblcost}\label{thm:lblcost}
The $\lbl$ algorithm can be implemented in $O(kn^4)$ time, where $k$ is the maximum number of levels.
\end{restatable}

In experiments we have run the $\lbl$ algorithm to termination and found that the number of levels is much smaller than $n$ and that it converges quickly.  Fig.~\ref{fig:nlvl2err} shows the convergence of the error as a function of the bound $k$ on the number of levels for $\W=\allrange(512)$.  In addition, although the worst case cost of the $\lbl$ algorithm is $O(kn^4)$, empirically we find that the running time increases approximately with $O(kn^3)$. 

\begin{figure}[t]
\centering
\includegraphics[width=350pt]{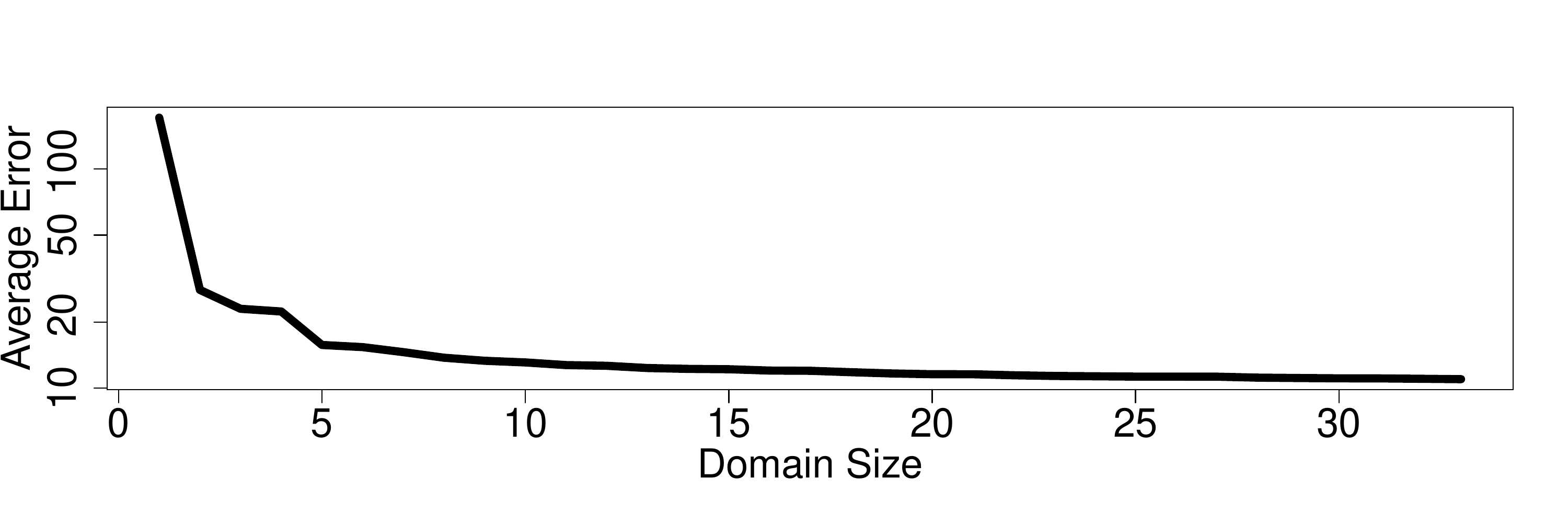}
\caption{The total error of strategies found by $\lbl$ as the bound $k$ on the number of levels increases.}
\label{fig:nlvl2err}
\end{figure}

\subsection{The LSA Algorithm on Large Domains}
In the remainder of this section we present two techniques for efficiently scaling the $\lbl$ algorithm to larger domains. 

\subsubsection{Workload Separation}

The domain for workloads over multi-dimensional data can increase with the product of the attribute domains.  To avoid this, many common workloads, such as those resulting from sets of low-order marginals, can be decomposed so that each dimension can be treated separately.  
This process of {\em workload separation} can significantly improve efficiency with little impact on overall error.  

\begin{definition}
Given a workload $\W$, query $\q$ is called a {\em separating query} for $\W$ if there exists a list of subsets $\W_1, \W_2, \ldots, \W_k$ of $\W$ for which the estimate of $\W_i$ does not involve queries in $\W_1, \ldots, \W_{i-1}, \W_{i+1}, \ldots, \W_k$ when the answer of query $\q$ is fixed. The subsets $\W_1, \W_2, \ldots, \W_k$ are called separated workloads of $\W$ under query $\q$.
\end{definition}

Given a workload $\W$ and one of its separating queries $\q$, notice that  answering queries in each separated workload does not involve other queries in the workload. \eat{; one can answer those separated workloads separately once the separating query $\q$ is answered.} Thus workload separation can be applied in two phases. In the first phase, the separating query $\q$ is answered directly, using the Gaussian mechanism with fixed accuracy. In the second phase, the strategies for each separated workload of $\W$ are computed with the $\lbl$ algorithm.  Since each separated workload contains queries with simpler structure, the buckets that always appear at the same time in one separated workload can be grouped into a single bucket so as to reduce the domain size.

One of the key applications of workload separation is for workloads consisting of sets of marginal queries and predicate queries over marginals, explained in the following theorem.

\begin{restatable}{theorem}{thmtotalsep}\label{thm:totalsep}
Let $\W_1, \ldots, \W_k$ be sets of predicate queries over one-way marginals on $k$ different dimensions. Let $\q_0$ be the sum of all the entries on the domain. With the answer of $\q_0$ given, the estimate of any query in $\W_i$ does not involve queries in $\W_1, \ldots, \W_{i-1}, \W_{i+1}, \ldots, \W_k$.
\end{restatable}

More generally, when a workload contains predicate queries over up to $k$-way marginals, such as datacube queries, workload separation can then be applied recursively, reducing the problem to low- or even one-dimensional problems.

\begin{example} 
Consider a query workload consisting of all one-dimensional range queries over each of two dimensions, with domains of size $n_1$ and $n_2$.  Such a workload would typically be represented over a set of cell counts of size $n_1 \times n_2$, and the $\lbl$ algorithm would take $O(k(n_1n_2)^4)$ time.  The separated workloads are $\allrange(n_1)$ and $\allrange(n_2)$ and the cost of running $\lbl$ is reduced to $O(k(n_1^4+n_2^4))$.
\end{example}

\subsubsection{Workload Generalization}

As a second optimization, called workload generalization, we initially merge cells in the domain, creating an approximation of the original workload over a much smaller domain.  The $\lbl$ algorithm is executed on this generalized workload to produce an initial generalized strategy.  Then, in a second phase, we run a slightly modified Program~\ref{alg:lbl} over each merged cell.  The modified algorithm computes error and sensitivity with the strategy of the first phase taken into account. The strategies on both phases are then combined to produce the final strategy.




The cost of this method is a function of domain size $n$ and the generalized domain size $m$. The first phase runs Prog~\ref{alg:lbl} once over domain size $m$ and the second phase runs modified Prog~\ref{alg:lbl} $m$ times over domain size $n/m$. Thus, the total cost of workload generalization is $O(k(m^4+{n^4}/{m^3}))$. If $m$ is set to $O(n^{1/3})$, then the total cost can be reduced to $O(kn^{3})$.  This is the asymptotic cost as matrix multiplication, thus it is a natural lower bound for the running time of an algorithm of this form.  Workload generalization results in much faster execution time, but sacrifices solution quality in comparison the $\lbl$ algorithm, resulting in modestly higher error.  We evaluate this experimentally in the next section.

\section{Experimental Evaluation}

There are two parts to the experimental evaluation of our techniques.  First, we study the complexity of different workloads using the singular value bound.  Second, we evaluate the quality of strategies generated by the $\lbl$ algorithm, as well as the effectiveness of the workload separation and generalization techniques.  Recall from Sec. \ref{sec:sub:error} that the error of the matrix mechanism, and therefore the choice of strategy, is independent of the database instance.  As a consequence, the results that follow do not use an input database---they are purely an analysis of workloads and the error rates possible in answering these workloads.  

\subsection{Workload Complexity}
We can use the singular value bound to gain insight into the complexity of workloads---the essential hardness of answering a set of queries under the matrix mechanism.  We begin by reporting the SVD bound for all range workloads on five different domains, all with a domain size $1024$ but different number of dimensions.
\begin{figure}[t]
	\centering
	\includegraphics[width=450pt]{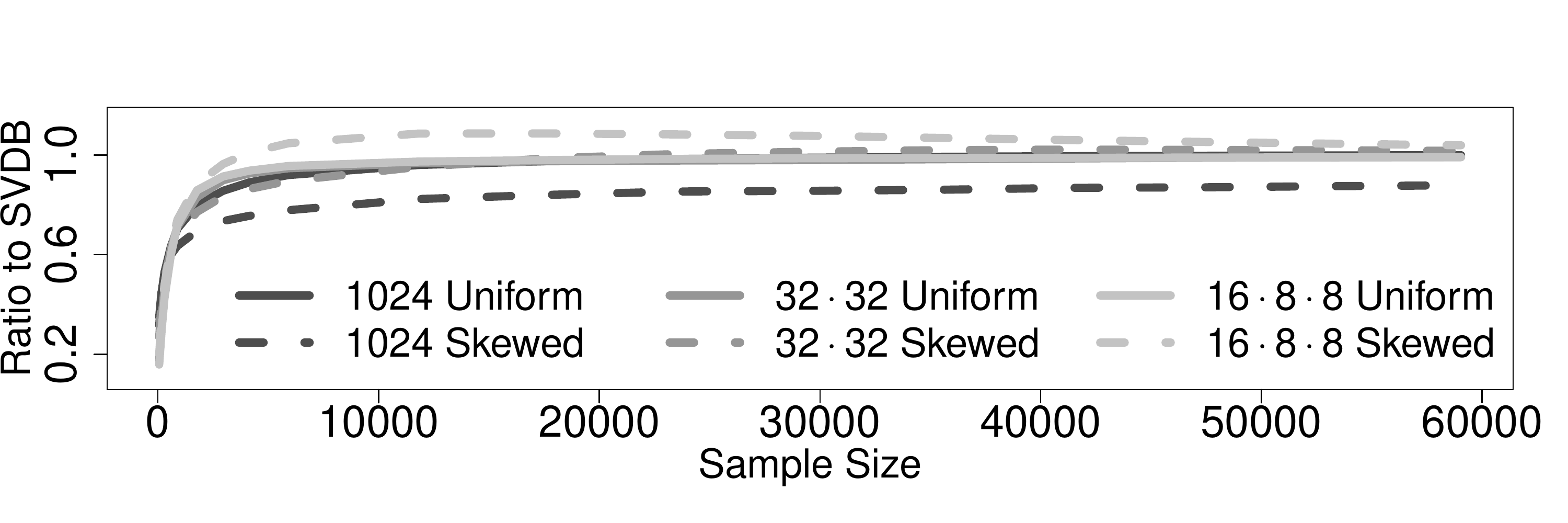}
	\caption{The SVD Bound for Sampled Workloads}
	\label{fig:sample}
\end{figure}
The comparison is shown as the lightest grey bars in {\bf Figure~\ref{fig:svd_lbl_h}}, where $2^{10}$ means the 10-dimensional domain with each dimension size equal to 2. As the number of dimensions increases with the overall domain fixed, the total number of range queries decreases, leading to lower complexity workloads, as shown by the decreasing SVD bound in the figure.

Our next experiment, shown in {\bf Figure~\ref{fig:sample}}, samples a matching number of distinct queries from each workload and considers the average error of the sampled query sets.  The complexity of a sampled query set is measured by the ratio between average error of the sample and the average error of the complete query set.  Moreover, the experiment uses both uniform sampling and a biased sampling method\footnote{\small Biased sampling chooses queries whose center is far away from the center of the domain, and the probability for choosing a range decreases exponentially with the distance between the center of the query and the center of the domain. The motivation of having this sampling method is to create a way of sampling that has different properties with uniform sampling, which can also be viewed as a process where users whose most interested region is at the extreme of the domain so that they tend to generate queries that are far away from the center.} to study whether the SVD bound is related to the sampling method.
\begin{figure*}[t]
	\centering
\subfigure[\small The $\lbl$ Algorithm]{
	\includegraphics[width=170pt]{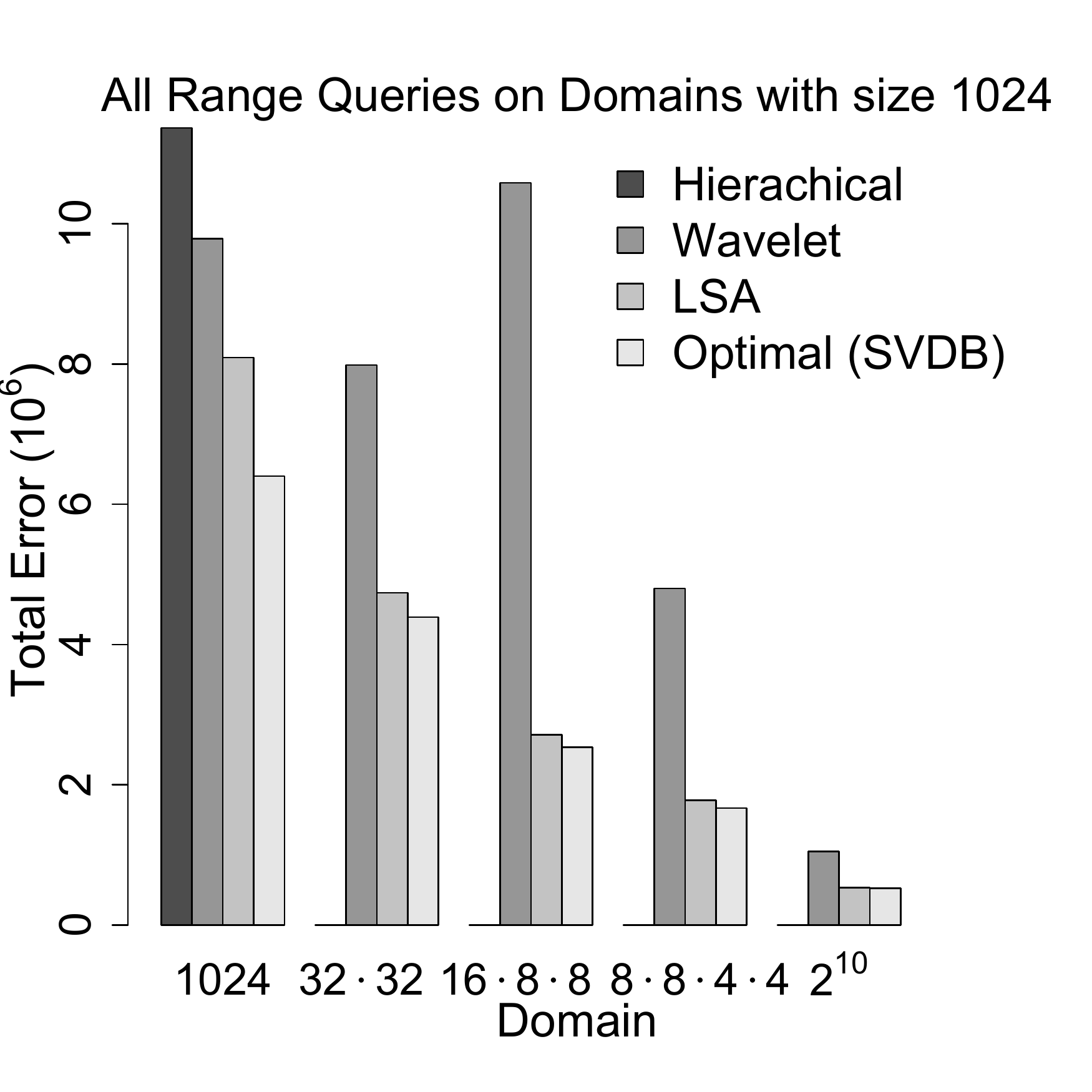}
	\label{fig:svd_lbl_h}}
\subfigure[\small Random Sampled Queries]{
	\includegraphics[width=170pt]{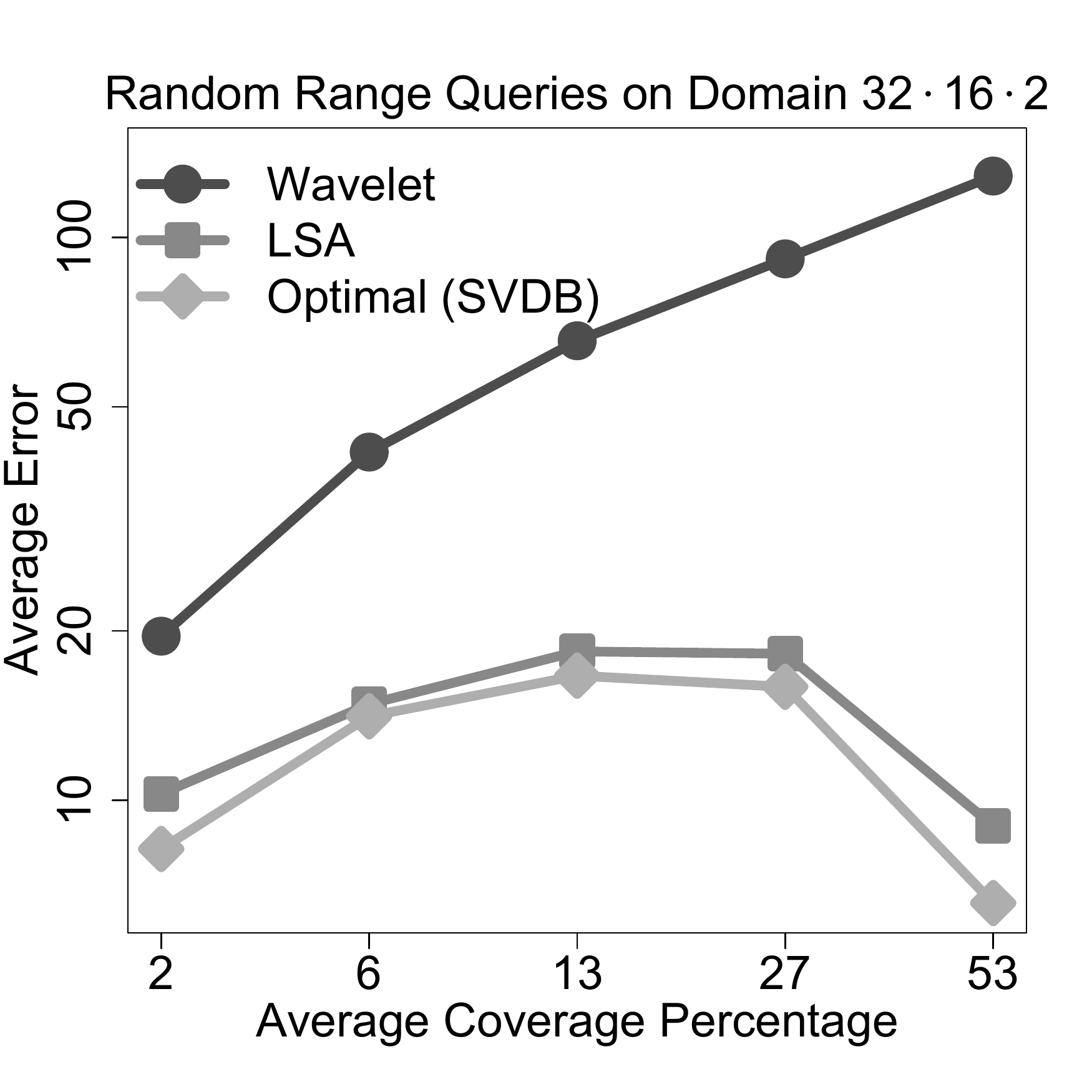}
	\label{fig:privsample}}\\
\subfigure[\small Workload Separation]{
	\includegraphics[width=170pt]{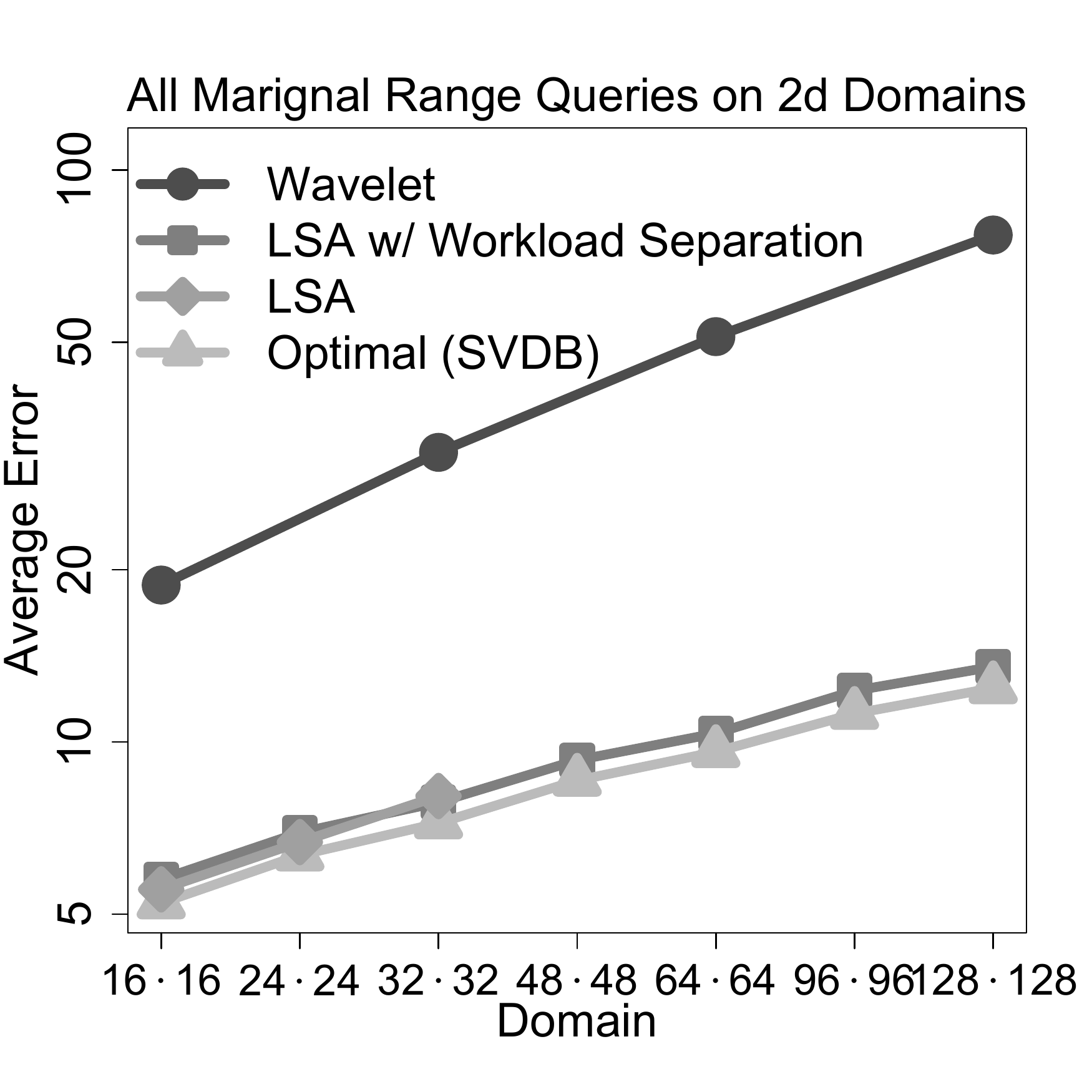}
	\label{fig:margin}}
\subfigure[\small Workload Generalization]{
	\includegraphics[width=170pt]{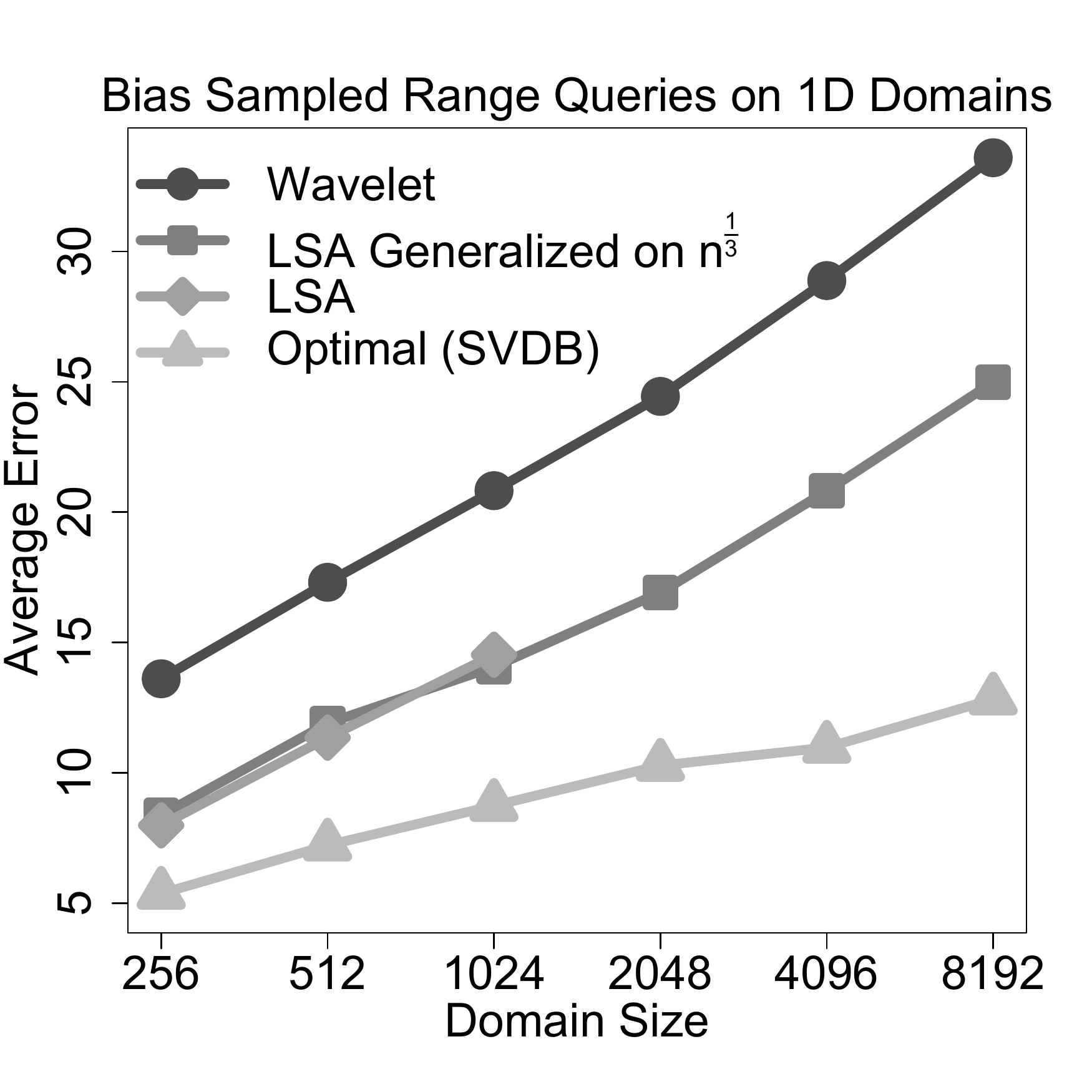}
	\label{fig:expand}}	
\caption{Performance of the $\lbl$ algorithm and auxiliary techniques.}
\end{figure*}
 
{\bf Figure \ref{fig:sample}} shows that even very small workloads of uniformly sampled queries (10000 queries, or roughly 2-6\% of the entire workload of range queries) reach the error levels of the entire workload.  When biased sampling is used to construct workloads, the error ratio approaches $1$ more slowly.  Overall, this may have important consequences for the strategy selection problem. It suggests that, for workloads of just modest size, there is no error penalty in adopting a query strategy tuned to a larger regular workload such as $\allrange$ for the appropriate domain size.

\subsection{Effectiveness of the $\lbl$ Algorithm}
To assess the effectiveness of the $\lbl$ algorithm we compare the total error under its strategies with other strategies from the literature, and with the optimal error as derived from the SVD bound.  The first experiment, shown in {\bf Figure~\ref{fig:svd_lbl_h}}, explores $\allrange$ workloads over the five multi-dimensional domains mentioned above.  We compare the $\lbl$-computed strategy with the wavelet strategy~\cite{xiao2010differential} and the hierarchical strategy~\cite{Hay:2010Boosting-the-Accuracy}.  Recall that both these strategy matrices were are originally designed for $\epsilon$-differential privacy, but in fact they perform even better under $(\epsilon,\delta)$-differential privacy.  The hierarchical tree algorithm only appears in the first group of comparisons because it is oriented towards one-dimensional queries. 

Both the wavelet and the hierarchical strategies perform well in the one-dimensional case: their error is $1.53$ times the optimal and $1.78$ times the optimal, respectively.  The $\lbl$ algorithm finds a better strategy, $1.26$ times the optimal.  In higher dimensions the performance improvement over wavelet is much larger, demonstrating the importance of adapting the strategy to the workload.  For example, for $\allrange(16,8,8)$, $\lbl$ error is $1.07$ times optimal while wavelet is $4.17$ times optimal.  Overall, the error of $\lbl$ strategies comes very close to the optimal singular value bound in higher dimensions.

To further test the benefits of adapting the strategy to the workload, we considered random workloads of range queries grouped by average coverage percentage, which is the average percentage of cells that are covered by a query. 
As shown in {\bf Figure~\ref{fig:privsample}} (note the logarithmic scale), the $\lbl$ strategies have error rates very close to optimal and outperform the wavelet strategy by more than an order of the magnitude in the most extreme case. 

Next we evaluate the workload separation and workload generalization techniques for improved efficiency of the $\lbl$ algorithm on large domains. {\bf Figure~\ref{fig:margin}} shows the workload separation technique on range queries over pairs of one-way marginals with increasing domain sizes.  Using $\lbl$ with workload separation has a negligible penalty in error: it is almost identical to the original $\lbl$ (and the optimal SVD bound), yet the computational cost is improved by more than three orders of magnitude: On domain $32\cdot32$, the running time is reduced from $5079$ seconds to merely $2.44$ seconds.

To evaluate the workload generalization techniques we considered skew-sampled workloads on one dimension, since their properties differ from $\allrange$ workloads.  {\bf Figure~\ref{fig:expand}} compares the performance of $\lbl$ generalized on domain size $n^{1/3}$ , the original $\lbl$ algorithm, the wavelet strategy, and the singular value bound.  Generalized $\lbl$ performs almost as well as the original $\lbl$ algorithm.  However, the computational cost is decreased significantly: for domain $1024$, the cost is reduced by a factor of 18. 
\section{Related Work}
The matrix mechanism \cite{Li:2010Optimizing-Linear} analyzed in a unified framework two prior techniques for accurately answering range queries.  The first used a wavelet transformation \cite{xiao2010differential}; the second used a hierarchical set of queries followed by inference \cite{Hay:2010Boosting-the-Accuracy}.  The original work on the matrix mechanism focused primarily on $\epsilon$-differential privacy, although $(\epsilon,\delta)$-differential privacy was considered briefly.  In both cases, it was shown that semi-definite programming could be used to compute a minimum error strategy for a workload, but solving such programs is not infeasible.  The present work provides tractable methods for computing low-error strategies for any given workload.

Workload complexity under $\epsilon$-diff\-erential privacy has been studied before. Hardt and Talwar~\cite{Hardt:2010On-the-Geometry-of-Differential} present a lower bound for randomly generated predicate workloads.  More generally, Blum et al.~\cite{blum2008a-learning} show that workload complexity is related to the VC dimension of the workload.  This measure is not directly comparable to our singular value bound since the former concerns $\epsilon$-differential privacy.  Also, the VC dimension of any one-dimensional range workload is constant, while our measure captures more detailed differences in workloads.

Put in our terms, Xiao et al.~\cite{xiaodifferentially} propose a method for computing a strategy matrix using KD-trees and private accesses to the database. Their query answers could be improved and made consistent by employing the matrix mechanism.  Roth and Roughgarden~\cite{Roth:2010The-Median-Mechanism:} describe a data-dependent mechanism to answer predicate queries on databases with 0-1 entries. Hardt et. al~\cite{hardt2010multiplicative} provide a linear time algorithm for the same query and database setting.  The trade-off between accuracy and efficiency among data-dependent and data-independent methods deserves further investigation.

\section{Conclusion}

Standard differentially-private mechanisms for answering a workload of queries require noise determined by the sensitivity of the queries.  We have shown that it is possible to satisfy the privacy condition with substantially less noise, and that the noise required is related to the spectral properties workload.  Our methods allow the privacy mechanism to be efficiently adapted to the workload, achieving error improvements of as much as one order of magnitude over prior techniques.


\bibliographystyle{abbrv} 
{ 
\bibliography{paper} 
}

\newpage
\appendix


\section{Data model and Domain} \label{app:datamodel}

In this section we provide a brief concrete example to explain how a workload of counting queries can be expressed as a set of linear queries over a vector of cell counts $\x$.  Consider the following relational schema describing students:
$$R=(name, gradyear, gender, gpa)$$
and suppose $dom(gradyear)=\{2011,2012,2013,2014\}$ and $dom(gender)=\{M,F\}$.  
If the desired workload consists of statistics about the gender of students 
graduating in specified years, then we can define the cells of $\x$ as the crossproduct of $dom(gradyear)$ and $dom(gender)$, which has size $n=8$.  That is, 
{\small $$\x=[ cnt(2011,M), cnt(2011,F), \dots cnt(2014,M), cnt(2014,F)]$$}
Then the following matrix represents a workload of five queries:
\begin{eqnarray*}
\W & = & 
\begin{bmatrix}
1 & 1 & 1 & 1 & 1 & 1 & 1 & 1 \\
1 & 1 & 1 & 1 & 0 & 0 & 0 & 0 \\
0 & 1 & 0 & 1 & 0 & 0 & 0 & 0 \\
1 & 0 & 1 & 0 & 0 & 0 & 0 & 0 \\
0 & 0 & 0 & 0 & 1 & 1 & \mbox{-}1 & \mbox{-}1 \\
\end{bmatrix}
\end{eqnarray*}
By rows, the queries represented by $\W$ are: $Q_1$: the count of all students; 
$Q_2$: the count of students with $gradyear \in [2011,2012]$;
$Q_3$: the count of female students with $gradyear \in [2011,2012]$;
$Q_4$: the count of male students with $gradyear \in [2011,2012]$;
$Q_5$: the difference between 2013 grads and 2014 grads.

Note that each query of interest should be listed in the workload.  We do {\em not} omit queries whose answers could be calculated from other queries in the workload.  For example, $Q_2$ is included in the workload even though it could be computed by summing $Q_3$ and $Q_4$.  Noise accumulates when summing noisy query answers, so we include each desired query and minimize the total error. Also, the relative accuracy of a query in the workload can be controlled by linearly scaling its row by a coefficient.

For this domain, the workload of all range queries is written $\allrange(4,2)$ and consists of all two dimensional range queries over $gradyear$ and $gender$, where the possible ``ranges'' for $gender$ are simply $M$, $F$, or $(M \vee F)$.  This workload consists of $30$ queries.  The workload $\allpred(8)$ includes all $2^8$ linear queries expressed over $\x$ with coefficients in $\{0,1\}$.

Also note that for the workload $\W$ above, the pair of elements $x_5,x_6$ and $x_7,x_8$ always appear together.  The included queries do not distinguish between males and females in year 2013, or in year 2014.  For this reason, each pair of variables could be replaced by a single variable, reducing the vector $\x$ to size 6.

\section{Comparison of the {\large $\epsilon$}- and  {\large $(\epsilon,\delta)$}-Matrix Mechanisms} \label{app:l1l2}

In this appendix we consider how our main results, which apply to approximate differential privacy, compare to corresponding results under standard differential privacy.  We begin with definitions for the $\epsilon$-matrix mechanism in Sec. \ref{app:l1l2:def}.  Then, in Sec. \ref{app:l1l2:diff}, we describe key differences that make the strategy selection problem harder under $\epsilon$-differential privacy.   While the privacy guarantees of the two mechanisms are formally distinct, for conservative settings of $\delta$, one may be indifferent to the two guarantees and consider which mechanism offers lower error for a fixed $\epsilon$.  In Sec. \ref{app:l1l2:error} we show that the error of the $(\epsilon,\delta)$-matrix mechanism is favorable for a wide range of strategies including those considered in this paper.

\subsection{Definitions, {\large $\epsilon$}-Matrix Mechanism} \label{app:l1l2:def}

Standard differential privacy is defined as follows:

\begin{definition}[Differential Privacy] A randomized algorithm $\alg$ is $\epsilon$-differentially private if for any instance $I$, any $I' \in \nbrs(I)$, and any subset of outputs $S \subseteq Range(\alg)$, the following holds:
\[
Pr[ \alg(I) \in S] \leq \exp(\epsilon) \times Pr[ \alg(I') \in S],
\]		
	\end{definition}

Under $\epsilon$-differential privacy, query sensitivity is measured using the $L_1$ distance.  For a query matrix $\W$, the $L_1$ sensitivity is the maximum $L_1$ norm of the columns of $\W$.  

\begin{proposition}[$L_1$ Query matrix sensitivity]
The $L_1$ sensitivity of a query matrix $\W$ is denoted $\Lone{\W}$ and is defined as follows:
\begin{eqnarray*}
\Lone{\W} & \eqbydef & \max_{\x' \in \nbrs(x)} \Lone{\W\x - \W\x'} 
			= \max_{W_i \in \cols(\W)} \Lone{W_i} \\
\end{eqnarray*}
\end{proposition}

The standard mechanism for achieving $\epsilon$-differential privacy adds Laplace noise calibrated to the $L_1$ sensitivity.  We use $\Lap(b)^m$ to denote a column vector consisting of $m$ independent samples drawn from a Laplace distribution with mean $0$ and scale $b$.  

\begin{proposition}[Laplace mechanism] \label{prop:laplace}
Given an $m \times n$ query matrix $\W$, the randomized algorithm $\LM$ that outputs the following vector is $\epsilon$-differentially private:
$$\LM(\W,\x) = \W\x + \Lap(\frac{\Lone{\W}}{\epsilon})^m$$
\end{proposition}

The matrix mechanism is defined almost identically to Prop. \ref{thm:l2diffpriv}, but with Laplace noise in place of Gaussian noise.

\begin{proposition}{\sc ($\epsilon$-Matrix Mechanism \cite{Li:2010Optimizing-Linear})} \label{def:e-m-mech} 
Let $\A$ be a full rank $m \times n$ strategy matrix and let $\W$ be any $p \times n$ workload matrix. Then the randomized algorithm $\MM_\A$ that outputs the following vector is $\epsilon$-differentially private:
\begin{eqnarray*}
\MM_\A(\W,\x) &=& \W\x + \W \A^\plus \Lap(b)^m.
\end{eqnarray*}
where $b=\Lone{\A}/\epsilon$
\end{proposition}

The analysis of error for the $\epsilon$-matrix mechanism differs only in the sensitivity and $\epsilon$ terms:

\begin{proposition}{\sc (Total Error)} \label{prop:totalerrorL1}
Given a workload $\W$, the total error of answering $\W$ using the $\epsilon$ matrix mechanism with query strategy $\A$ is:
\begin{equation}\label{eqn:totalerrorL1}
 \error\A{\W} = \frac{2}{\epsilon^2}|| \A ||_1^2 \;\tr (\W^T\W(\A^T\A)^{-1})
 \end{equation}
\end{proposition}

\begin{figure}[t]
\centering 
 
\begin{tabular}{cccc} 
$\left[\begin{matrix}
1 & 1 & 1 & 1 \\
1 & 1 & 0 & 0 \\
0 & 0 & 1 & 1 \\
1 & 0 & 0 & 0\\
0 & 1 & 0 & 0\\
0 & 0 & 1 & 0 \\
0 & 0 & 0 & 1 \\
\end{matrix}\right]$ 
 \hfill &  
$\left[\begin{matrix}
1 & 1 & 1 & 1\\
1 & 1 & \mbox{-}1 & \mbox{-}1\\
1 & \mbox{-}1 & 0 & 0 \\
0 & 0 & 1 & \mbox{-}1 \\
\end{matrix}\right]$ 
\hfill&
$\left[\begin{matrix}
1 & 1 & 0 & 0 \\
0 & 0 & 1 & 1 \\
1 & 0 & 0 & 0\\
0 & 1 & 0 & 0\\
0 & 0 & 1 & 0 \\
0 & 0 & 0 & 1 \\
1 & 0 & 0 & 0\\
0 & 1 & 0 & 0\\
0 & 0 & 1 & 0 \\
0 & 0 & 0 & 1 \\
\end{matrix}\right]$ 
\hfill&
 $\left[\begin{matrix}
1 & 1 & 0 & 0 \\
0 & 0 & 1 & 1 \\
\sqrt{2} & 0 & 0 & 0\\
0 & \sqrt{2} & 0 & 0\\
0 & 0 & \sqrt{2} & 0 \\
0 & 0 & 0 & \sqrt{2} \\
\end{matrix}\right]$ \\
\\
$\H$ & $\Wav$ & $\Wav_1$ & $\Wav_2$ \\
\end{tabular}
\caption{\label{fig:query-matrix} Strategy matrices for domain $n=4$.  $\H$ is the hierarchical strategy.  $\Wav$ is the wavelet strategy.  $\Wav_1$ has redundant queries which are reduced in $\Wav_2$.  $\Wav,\Wav_1,\Wav_2$ are all equivalent under the $(\epsilon,\delta)$-matrix mechanism.  Under the $\epsilon$-matrix mechanism, $\Wav_2$ is strictly more efficient than $\Wav_1$ because $\Lone{\Wav_2} < \Lone{\Wav_1}$.}
\end{figure}

\subsection{Main Distinctions} \label{app:l1l2:diff}

The analysis of error differs between the two mechanisms because of the difference in sensitivity metrics.  From Prop. \ref{prop:totalerrorL1} and Prop. \ref{prop:totalerror} the two expressions for error can be compared (ignoring privacy parameters):
\begin{eqnarray*}
\epsilon\mbox{-}\error{\A}{\W} & \propto & || \A ||_1^2 \;\tr (\W^T\W(\A^T\A)^{-1}) \\
(\epsilon,\delta)\mbox{-}\error{\A}{\W} & \propto & || \A ||_2^2 \;\tr (\W^T\W(\A^T\A)^{-1})\end{eqnarray*}
The trace terms are identical, so these expressions differ only in the sensitivity metric applied to $\A$.  The implications of this small difference are significant for optimizing the expressions, primarily because $\Ltwo{\A}$ is uniquely determined by $\A^T\A$ (it is in fact the largest diagonal entry of matrix $\A^T\A$). This means that whenever two different strategies $\A$ and $\B$ are such that $\A^T\A=\B^T\B$, then under $(\epsilon, \delta)$-matrix mechanism they have equivalent error, so $\A \equiv \B$.  But $\Lone{\A}$ is not determined by $\A^T\A$, so it is possible to have $\A^T\A=\B^T\B$, while $\Lone{\B} < \Lone{\A}$.  Then it follows that under $\epsilon$-matrix mechanism $\A$ and $\B$ are not equivalent.  Instead strategy $\B$ is strictly more efficient than $\A$.  


This difference in invalidates many results that hold for the $(\epsilon,\delta)$-matrix mechanism. First of all, the optimal strategy under $(\epsilon,\delta)$-matrix mechanism can be found via a semi-definite programming (SDP) that computes the matrix $\A^T\A$ to minimize the total error formula.  The solution to this program is insufficient, however, under the $\epsilon$-matrix mechanism, because we need to compute an strategy $\A$ with minimum $L_1$ sensitivity.  To do so, extra non-convex constraints must be added which requires solving a non-convex extension of an SDP (an SDP with rank constraints) to get an optimal strategy. 

Further, the singular value bound does not hold under the $\epsilon$-matrix mechanism and the properties of redundant queries under $\epsilon$-matrix mechanism change. Unlike Thm.~\ref{thm:redqueries}, the presence of redundant queries leads to unneeded error:

\begin{theorem}[Redundant Queries]\label{thm:redqueriesL1}
Suppose strategy $\A_1 = \{\A_0 \cup \q \cup c_1\q\}$ for some strategy $\A_0$, some linear query $\q$, and some constant $c_1$.  Then,  under $\epsilon$-matrix mechanism,  the reduced strategy $\A_2 = \{\A_0 \cup c_2\q\}$ where $c_2 = \sqrt{1+c_1^2}$ is strictly more efficient than $\A_1$.
\end{theorem}

Fig. \ref{fig:query-matrix} shows an example of a strategy matrix with redundant queries ($\Wav_1$) and a reduced strategy with lower $L_1$ sensitivity.  

Finally, column uniformity no longer characterizes the minimal strategies (as in Thm.~\ref{thm:columnunif}) and the quality of the output of the $\lbl$ algorithm tends to be lower.  Intuitively, in the $\epsilon$-matrix mechanism, the algorithm only has one chance to add a query to the strategy and needs to choose the correct weight before knowing how the remainder of the strategy will be built.


\subsection{Error Comparison} \label{app:l1l2:error}

Suppose we fix $\epsilon$ in both mechanisms and choose $\delta=2/n^2$ in $(\epsilon,\delta)$-matrix mechanism, where $n$ is the domain size.  This is a conservative choice for $\delta$, which results in error of the $(\epsilon,\delta)$-matrix mechanism equal to:
\[\error{\A}{\W}=\frac{2}{\epsilon^2}(\log^{\frac{1}{4}}n ||A||_2)^2\tr(\W^T\W(\A^T\A)^{-1}).\]
Comparing this equation with Eq.~(\ref{eqn:totalerrorL1}), indicates that the $(\epsilon, \delta)$ mechanism introduces less error whenever $||\A||_1 > \log^{\frac{1}{4}}n||\A||_2$. In particular, for strategy matrices $\A$ which consists of predicate queries (such as the output of $\lbl$ algorithm, hierarchical strategies \cite{Hay:2010Boosting-the-Accuracy} and (a strategy equivalent to) the wavelet strategy \cite{xiao2010differential}), $||\A||_2 = \sqrt {||\A||_1}$.  Using such a strategy matrix $\A$, when $||\A||_1>\sqrt{\log n}$, 
the $(\epsilon,\delta)$-matrix mechanism provides less error than $\epsilon$-matrix mechanism. 
\section{The Optimal Strategy for Variable Agnostic Workloads}\label{app:varagnostic}

To demonstrate that the singular value bound is a tight lower bound, we discuss a special type of workloads called the variable agnostic workloads and construct a strategy to such workloads that introduce error as low as the singular value bound.
\begin{definition}[Variable agnostic workload]
A\\workload $\W$ is {\em variable agnostic} if $\W^T\W$ is unchanged when we swap two columns of $\W$.
\end{definition}

For any variable agnostic workloads $\W$, $\W^T\W$ has the following form, for constants $a$ and $b$:

The next corollary gives a strategies which work on variable agnostic workloads with size $n$ and have the total error equal to the the singular value bound.

\begin{restatable}{theorem}{thmvaragn} \label{thm:varagn}
For positive integer $k$ and $n=2^k$, let $\W$ be any $m\times n$ variable-agnostic workload.  Then $\W^T\W$ has the following form, for constants $a$ and $b$:
\[\small
\left[\begin{array}{cccc}
a & b & \ldots & b\\
b & a & \ldots & b\\
\vdots & \vdots & \ddots & \vdots \\
b & b & \ldots & a
\end{array}\right].
\]
and there exists a strategy $\A$ which attains the singular value bound which is $\svdb(\W)=\frac{1}{n}(\sqrt{a+(n-1)b}+(n-1)\sqrt{a-b})^2$.
\end{restatable}
\begin{proof}
Let $\Q_1=\left[\begin{smallmatrix}1 & 1 \\1 & -1\end{smallmatrix}\right]$ and $\Q_k=\left[\begin{smallmatrix}\Q_{k-1} & \Q_{k-1} \\ \Q_{k-1} & -\Q_{k-1}\end{smallmatrix}\right]$. Here we prove the following result: $Q_{k}$ is an eignvector matrix for $\W_k^T\W_k$ where $\W_k$ is any $2^{k}\times 2^{k}$ variable-agnostic workload. The eigenvalue corresponds to the first column of $\Q_{k}$ is $a+(2^k-1)b$ and the eigenvalue corresponds to all the other columns is $a-b$, where $a$ is the diagonal entry of $\W_k^T\W_k$ and $b$ is the off diagonal entry of $\W_k^{T}\W_k$.

One can verify the result is true when $k=1$. Now suppose the result is true for $k-1$. Then
{\small
\begin{eqnarray}
&&\W^T\W\Q_k=\left[\begin{array}{cc}\W_{k-1} & b\mathbf{J}_{2^{k-1}} \\b\mathbf{J}_{2^{k-1}} & \W_{k-1}\end{array}\right]\left[\begin{array}{cc}\Q_{k-1} & \Q_{k-1} \\ \Q_{k-1} & -\Q_{k-1}\end{array}\right]\nonumber \\
&&=\left[\begin{array}{cccc}
a+(2^k-1)b & 0 & \ldots & 0\\
0 & a-b & \ldots & 0\\
\vdots & \vdots & \ddots & \vdots \\
0 & 0 & \ldots & a-b
\end{array}\right].
\Q_{k}\label{eqn:prfeigenmat}
\end{eqnarray}}
Here $\mathbf{J}$ is the matrix whose all entries are $1$. The computation of Eq.~\ref{eqn:prfeigenmat} uses the fact that the sum of first column of $\Q_k$ is $2^k$ and the sum of any other column of $\Q_k$ is $0$. In addition, since, 
\[\small
\Q_k^T\Q_k=\left[\begin{array}{cc}\Q_{k-1} & \Q_{k-1} \\ \Q_{k-1} & -\Q_{k-1}\end{array}\right]
\left[\begin{array}{cc}\Q_{k-1} & -\Q_{k-1} \\ \Q_{k-1} & \Q_{k-1}\end{array}\right]=\frac{1}{2^k}\I,\]
we know $\Q_k$ is a eigenvector matrix for $\W^T\W$.

From above, we know the eigenvalues and a set of corresponding eigenvectors of $\W^T\W$. According to the proof of Thm.~\ref{thm:singularvaluebound}, the $\svdb(\W)$ can be achieved if and only if 
\[\small
\left[\begin{array}{cccc}
\sqrt{a+(2^k-1)b} & 0 & \ldots & 0\\
0 & \sqrt{a-b} & \ldots & 0\\
\vdots & \vdots & \ddots & \vdots \\
0 & 0 & \ldots & \sqrt{a-b}
\end{array}\right].
\Q_{k}\]
is a column-uniform matrix. Notice $Q_k^s=\mathbf{J}_{2^k}$, according to Thm.~\ref{thm:columnnormuniform}, the matrix above is column-uniform. 
\end{proof}

The total errors of using the identity matrix or the workload itself as the strategy matrix are both $na$. Compared with those total errors,  the ratio of total error reduced by using the strategy in Thm.~\ref{thm:varagn} is approximately $1-\frac{b}{a}$.

Since the workload $\allpred(n)$ is variable agnostic, a consequence of the above theorem is that we can find its optimal strategy.

\begin{corollary}
For $n=2^k$, the minimized total error for the workload $\allpred(n)$ is equal to the singular value bound, which is $\frac{2^{n-2}}{n}(n-1+\sqrt{n+1})^2$.
\end{corollary}


\section{Proof of Main Theorems} \label{app:theory}

This section contains proofs of the most important theorems in Section~\ref{sec:theory} and \ref{sec:algorithm}.

\subsection{Analysis of Minimal Strategies}\label{app:propstrategy}
Here we complete the proof characterizing minimal strategies and continue the discussion in order to prove the singular value bound.
\thmcolumnunif*
\begin{proof}
If a strategy $\A_1$ is not column-uniform, the strategy $\A_2$ that is more efficient $\A_1$ can be found according to the augmenting strategy theorem from \cite{Li:2010Optimizing-Linear}. Here we only show the proof that any column-uniform matrix is minimal under the partial order.

Given two column uniform strategies $\A_1$ and $\A_2$ such that $\A_1\leq \A_2$. Without loss of generality, we assume both of them have $L_2$ sensitivity 1, which means all diagonal entries of $\A_1^T\A_1$ and $\A_2^T\A_2$ are 1. Since $\A_1\leq \A_2$, for any query $\w$, $\error{\A_1}{\w} \leq \error{\A_2}{\w}$. According to the definition, $\error{\A_1}{\w} = \w^T\A_1^T\A_1\w$, $\error{\A_2}{\w}= \w^T\A_2^T\A_2\w$. Therefore, 
\begin{eqnarray}
\error{\A_1}{\w} - \error{\A_2}{\w} &= &\w^T\A_1^T\A_1\w - \w^T\A_2^T\A_2\w \nonumber\\
&=& \w^T(\A_1^T\A_1 - \A_2^T\A_2)\w \nonumber\\
&\leq& 0 \label{eqn:columnunif}
\end{eqnarray}
Since (\ref{eqn:columnunif}) is true for arbitrary $\w$, $\A_2^T\A_2-\A_1^T\A_1$ is a positive semi-definite matrix. In addition, as the assumption, the diagonal entries of $\A_1^T\A_1$ and $\A_2^T\A_2$ are 1. Therefore the diagonal entries of  $\A_2^T\A_2-\A_1^T\A_1$ are all 0. According to properties of positive semi-definite matrix, there is an unique semi-definite matrix whose diagonal entries are all 0, which is the 0 matrix. Thus we have $\A_1 = \A_2$. 
\end{proof}

As an application of Thm.~\ref{thm:columnunif}, we have the following theorem.
\begin{theorem}\label{thm:columnnormuniform}
If an $m \times n$ strategy matrix $\A$ connects to an optimal solution of Problem~\ref{prob:mintotal}, $\A\in\U_{m\times n}$, which is also equivalent to
{\small
\begin{equation}\label{eqn:columnnormuniform}
\P_\A^s\lambdaB_\A^s=k\left[\begin{array}{c}1\\1\\ \vdots\\1\end{array}\right].
\end{equation}
}
Here $\A=\Q_\A\LambdaB_\A\P_\A$ is the singular decomposition of $\A$, vector  $\lambdaB_\A$ is the vector consists of the singular values of $\A$, $k$ is a real number and $\A^s$ represents the element-wise square of a matrix $\A$.
\end{theorem}
For any strategy matrix $\A$, one can find an $n\times n$ matrix $\B$ such that $\A^T\A=\B^T\B$ by decomposing matrix $\A^T\A$. It is follows that $\error\A{\W}=\error\B{\W}$ and therefore it is sufficient to consider only $n\times n$ strategy matrices in Problem~\ref{prob:mintotal} .Thus (\ref{eqn:totalerror}) can be represented using Frobenius norms, which is defined as following.
\begin{definition} {\sc (Frobenius Norm)}
Given an $m\times n$ matrix $\A=\{a_{ij}\}$, the Frobenius norm of matrix $\A$ is defined as the square root of the sum of squares of its elements
\[||\A||_F=\sqrt{\sum_{i=1}^m\sum_{j=1}^n a_{ij}^2}.\]
It is also equal to the square root of the trace of matrix $\A^T\A$
\[||\A||_F=\sqrt{\tr(\A^T\A)}.\]
\end{definition}
According to the definition of Frobenius norm, assume strategy matrix $\A$ is a square matrix,  (\ref{eqn:mintotalerror}) can be rewritten as follows:
\begin{eqnarray*}
&&\min_\A|| \A ||_2^2\tr ((\A^T\A)^{-1}\W^T\W)\\
& =&\min_\A|| \A ||_2^2\tr (\W\A^{-1}(\A{^{-1}})^{T}\W^T)\\
&=&\min_\A||\A||_2^2||\W\A^{-1}||_F^2.
\end{eqnarray*}

\begin{lemma}\label{lem:independentqueries}
If a workload $\W$ consists of two sets of independent queries, i.e. there exists an orthogonal matrix $\P$ such that 
\[\small
\P\W=\left[
\begin{array}{cc}
\W_1 & 0\\
0 & \W_2
\end{array}
\right],\]
the strategy matrix $\A$ that minimizes total error has form 
\[\A=\left[
\begin{array}{cc}
\A_1 & 0\\
0 & \A_2
\end{array}
\right],\]
where $\A_1$, $\A_2$ are the solutions to problem~\ref{prob:mintotal} for $\W_1$, $\W_2$, respectively.
\end{lemma}

\begin{proof}
Given a strategy matrix $\A$, consider the cholesky decomposition of $\A^T\A$, which is in form 
\[\small
\A'=\left[
\begin{array}{cc}
\A_1 & \A_3\\
0 & \A_2
\end{array}
\right].\]
Then
{\small
\begin{eqnarray*}
\W\A'^{-1}&=&\left[\begin{array}{cc}
\W_1 & 0\\
0 & \W_2
\end{array}
\right]\left[
\begin{array}{cc}
\A_1^{-1} & -\A_1^{-1}\A_3\A_2^{-1}\\
0 & \A_2^{-1}
\end{array}
\right]\\
&=&\left[\begin{array}{cc}
\W_1\A_1^{-1} & -\W_1\A_1^{-1}\A_3\A_2^{-1}\\
0 & \W_2\A_2^{-1}
\end{array}\right].
\end{eqnarray*}
}
Let strategy matrix $\A_d$ be
\[\small
\A_d=\left[
\begin{array}{cc}
\A_1 & 0\\
0 & \A_2
\end{array}
\right].\]
Notice that
\begin{eqnarray*}
||\W\A'^{-1}||_F^2&=& ||\W_1\A_1^{-1}||_F^2 + ||\W_2\A_2^{-1}||_F^2 + ||\W_1\A_1^{-1}\A_3\A_2^{-1}||_F^2\\
&\leq&||\W_1\A_1^{-1}||_F^2 + ||\W_2\A_2^{-1}||_F^2 = ||\W\A_d^{-1}||_F^2\\
||\A'||_2^2 &\leq& \max\{||\A_1||_2^2, ||\A_2||_2^2\} = ||\A_d||_2^2,
\end{eqnarray*}
we know $\error{\A_d}\W \leq \error{\A}\W$. Thus the strategy matrix that minimizes (\ref{eqn:totalerror}) also has the same form as $\A_d$. Furthermore, if $\A_1$ is not an optimal strategy for workload $\W_1$, let $\A_1^*$ be an optimal strategy for $\W_1$. Notice $\frac{||A_1||_2}{||A_!^*||_2}\A_1^*$ is also an optimal strategy for $\W_1$ and substitute $\A_1$ by $\frac{||A_1||_2}{||A_!^*||_2}\A_1^*$ in $\A_d$ will bring down the Frobenius norm without boosting $||A_d||_2$ so that the total error becomes smaller. Thus, if $\A_d$ is an optimal strategy for $\W$, $\A_1$, $\A_2$ must be the solutions to problem~\ref{prob:mintotal} for $\W_1$, $\W_2$, respectively.
\end{proof}

\thmsvdbound*
\begin{proof}
For a given workload $\W$, according to Thm.~\ref{thm:columnnormuniform}, its optimal strategy matrix $\A$ satisfies $\A\in\U_{n\times n}$. Therefore,
\[||\A||_2=\frac{1}{n}\tr(A).\]
Let $\W=\Q_\W\LambdaB_\W\P_\W$ and $\A=\Q_\A\LambdaB_\A\P_\W$ be the singular decomposition of $\W$ and $\W$, respectively,  we have:
\begin{eqnarray}
&&\error{\A}{\W}\nonumber\\
&=&\min_{\A\in\U_{n\times n}}||\A||_2\tr(\W(\A^T\A)^{-1}\W^T)\nonumber\\
&=&\frac{1}{n}\min_{\A\in\U_{n\times n}}\tr(A)\tr((\A^T\A)^{-1}\W^T\W)\nonumber\\
&=&\frac{1}{n}\min_{(\LambdaB_\A\P_\A)\in\U_{n\times n}}\tr(\LambdaB_\A^2)\nonumber\\
&\cdot&\tr(\P^T_\A(\LambdaB_\A^{-1})^2\P_\A\P_\W^T\LambdaB_\W^2\P_\W)\nonumber\\
&\geq&\frac{1}{n}\min_{\LambdaB_\A, \P_\A}\tr(\LambdaB_\A^2)\nonumber\\
&\cdot&\tr(\LambdaB_\W(\P_\W^T\P_\A)^T(\LambdaB_\A^{-1})^2(\P_\W^T\P_\A)\LambdaB_\W)\label{eqn:minap}
\end{eqnarray}
As the proof of Lemma~\ref{lem:independentqueries}, the minimum of (\ref{eqn:minap}) achieved when $\P_\W^T\P_\A$ is a diagonal matrix. Thus $\P_\A=\P_\W$ and
\begin{eqnarray}
(\ref{eqn:minap}) &=& \frac{1}{n} \min_{\LambdaB_\A}\tr(\LambdaB_\A^2)\tr(\LambdaB_\W^2(\LambdaB_\A^{-1})^2)\nonumber \\
&\geq&\frac{1}{n}(\sum_{i=1}^n\lambdaB_i)^2.\label{eqn:jensen}
\end{eqnarray}
The equal sign in (\ref{eqn:jensen}) is satisfied if and only if $\LambdaB_\A=\sqrt{\LambdaB_\W}$. Notice the inequality in (\ref{eqn:minap}) comes from removing the constraint that $(\LambdaB_\A\P_\A) \in \U_{n\times n}$, to satisfied the equal signs in (\ref{eqn:minap}) and (\ref{eqn:jensen}) simultaneously, we need $\sqrt{\LambdaB_\W}\P_\W \in \U_{n\times n}$. 
\end{proof}
\eat{
\subsection{Properties of Singular Value Bound}\label{app:svdprop}
The relationship between contained workloads and their singular value bound is discussed in this part. To prove Thm.~\ref{thm:containedwkld}, we need to proof the following lemma first.
\begin{lemma}\label{lem:diagmin}
Let $\D$ be a diagonal matrix with positive diagonal entries and $\P$ be a orthogonal matrix whose column equals to $\p_1,\p_2,\ldots,\p_n$. 
\[\tr(\D)\leq\sum_{i=1}^n||\D\p_i||_2.\]
\end{lemma}
\begin{proof}
Recall (\ref{eqn:minap}) in the proof of Thm.~\ref{thm:singularvaluebound}. On one hand,
we have:
\begin{eqnarray*}
(\ref{eqn:minap})&=&\frac{1}{n}\min_{\LambdaB_\A, \P_\A}\tr(\LambdaB_\A^2)\\
&\cdot&\tr(\LambdaB_\W(\P_\W^T\P_\A)^T(\LambdaB_\A^{-1})^2(\P_\W^T\P_\A)\LambdaB_\W)\\
&=&\frac{1}{n}\min_{\LambdaB_\A, \P}\tr(\LambdaB_\A^2)\tr(\LambdaB_\W\P^T(\LambdaB_\A^{-1})^2\P\LambdaB_\W)\\
&=&\frac{1}{n}\min_{\LambdaB_\A, \P, \tr(\A)=1}\tr((\LambdaB_\A^{-1})^2\P\LambdaB_\W^2\P^T)\\
&=&\frac{1}{n}\min_{\P}(\sum_{i=1}^n||\LambdaB_\W\p_i||_2)^2,
\end{eqnarray*}
where  $\p_1,\p_2,\ldots,\p_n$ are the columns of $\P$. On the other hand, from the proof of Thm.~\ref{thm:singularvaluebound}, we know
\[(\ref{eqn:minap})\geq\frac{1}{n}(\tr(\LambdaB_\W))^2.\]
Therefore, notice all entries of $\LambdaB_\W$ are positive, 
\[\tr(\LambdaB_\W)\leq\min_\P\sum_{i=1}^n||\LambdaB_\W\p_i||_2.\]
Since $\LambdaB_\W$ can be arbitrary diagonal matrix with positive diagonal entries, we have the lemma proved.
\end{proof}

Now we can continue our proof of Thm.~\ref{thm:containedwkld}.
\thmcontainedwkld*
\begin{proof}
Since $\W'_1$ contains all the rows of $\W_2$, let $\W'_1$ be the following form:
\[\W'_1=\left[\begin{array}{c}
\W_2 \\
\W_3
\end{array}\right].\]
Then
\begin{eqnarray*}
\W_1^T\W_1-\W_2^T\W_2 &=& {\W'_1}^T\W'_1-\W_2^T\W_2\\
&=&\W_3^T\W_3\succeq 0
\end{eqnarray*}
Let $\W_1=\Q_1\LambdaB_1\P_1$ and $\W_2=\Q_2\LambdaB_2\P_2$ be the singular value decomposition of $\W_1$ and $\W_2$, respectively. Then
\begin{eqnarray*}
\W_1^T\W_1-\W_2^T\W_2 \succeq 0&\Leftrightarrow& \P_1^T\LambdaB_1^2\P_1-\P_2^T\LambdaB_2^2\P_2 \succeq 0\\
&\Leftrightarrow&\LambdaB_1^2-\P_1\P_2^T\LambdaB_2^2\P_2\P_1^T \succeq 0\\
&\Leftrightarrow&\forall\,i,\,\lambdaB_i\geq ||\LambdaB_2\p_i||_2,
\end{eqnarray*}
where $\lambdaB_i$ is the $i$-th diagonal entry of $\LambdaB_1$ and $\p_i$ is the $i$-th column vector of $\P_2\P_1^T$. Therefore, according to Lemma~\ref{lem:diagmin},
\[\tr(\LambdaB_1)=\sum_{i=1}^n\lambdaB_i\geq\sum_{i=1}^n  ||\LambdaB_2\p_i||_2\geq \tr(\LambdaB_2).\]
\end{proof}
}
\subsection{The $\lbl$ Algorithm and Extensions.}\label{app:lbl}
Lastly, we prove the complexity of the $\lbl$ algorithm and the basis of the workload separation technique.
\begin{lemma}{\sc (Sherman-Morrison-Woodbury formula\cite{Hager:1989:UIM:75568.75570})}\label{lem:invertmat}
Given $n\times n$ matrix $\X$, $n\times m$ matrix $\mathbf{U}$, $m\times m$ matrix $\LambdaB$ and $m\times n$ matrix $\V$,
\begin{equation}\label{eqn:invertmat}
\inv{(\X-\mathbf{U\LambdaB}\V)} = \inv{\X} - \inv{\X}\mathbf{U}\inv{(\LambdaB+\mathbf{U}\inv{\X}\V)}\V\inv{\X}
\end{equation}
\end{lemma}

Using this lemma, the test of each split point in step~6 of Program~\ref{alg:lbl} can be done in $O(n^2)$ time, resulting in the following overall running time:
\thmlblcost*
\begin{proof}
Notice in step 6 of Program~\ref{alg:lbl}, each time we compute the total error of a possible split of row $[v_1, v_2, \ldots, v_n]$ on position $i$, it is equivalent to modify matrix $\A'$ by removing this selected row and add the following two rows to matrix $\A'$. 
\[\left[\begin{array}{cccc}v'_1&v'_2&\ldots&v'_n\\
v''_1 & v''_2 & \ldots & v''_n\end{array}\right],\quad v'_i, v''_i\in\{0,1\}, v'_i+v''_i=1,\, 1\leq i\leq n.\]

To apply Lemma \ref{lem:invertmat}, let $\X={\A'}^T\A'$ and $\mathbf{U}$, $\LambdaB$, $\V$ be the following matrices:
\[\mathbf{U}=\V=\left[\begin{array}{cccccc}v_1 &  v_2 & \ldots & v_n\\
v'_1&v'_2&\ldots&v'_n\\
v''_1 & v''_2 & \ldots & v''_n\end{array}\right]
,\LambdaB=\left[\begin{array}{ccc}-1 & 0 & 0 \\ 0 & 1 & 0 \\ 0 & 0 & 1\end{array}\right]\]

Denote the modified matrix as $\A''$ and we can verify that ${\A''}^T\A''={\A'}^T\A' - \mathbf{U}\LambdaB\V=\X-\mathbf{U}\LambdaB\V$. Therefore the inverse of ${\A''}^T\A''$ can be computed as the Eq.~(\ref{eqn:invertmat}). Since $\mathbf{U}$ and $\V$ are $n\times 3$ and $3\times n$ matrices respectively, the inverse of ${\A''}^T\A''$ can be computed in $O(n^2)$ time by first computing $\inv{\X}\mathbf{U}$, $\mathbf{U}\X\V$ and $\V\inv{\X}$ and then finishing the evaluation from left to right.
\end{proof}

\thmtotalsep*
\begin{proof}
For any $i$, $1\leq i\leq k$ and a query $\q\in\W_i$. To estimate $\q$, one needs to find a set of queries whose linear combination is equal to $\q_0$. Let the representation be
\[\q=\sum_{j=1}^l\alpha_j\q_j+\sum_{j=1}^h\beta_j\p_j,\]
where $\q_j\in\W_i$, $1\leq j\leq l$ and $\p_j\in\bigcup_{t=1, t\neq i}^n \W_t$, $1\leq j\leq h$. The equation above is equivalent to 
\begin{equation}\label{eqn:marginequ}
\q-\sum_{j=1}^l\alpha_j\q_j=\sum_{j=1}^h\beta_j\p_j.
\end{equation}
Notice the left hand side of Eq.~(\ref{eqn:marginequ}) is the sum of one-way marginal queries on $i$-th dimension and the right hand side of Eq.~(\ref{eqn:marginequ}) is the sum of one-way marginal queries that are not on $i$-th dimension. Since the only query that shared by the one-way marginal queries on $i$-th and not on $i$-th dimension is the total sum $\q_0$,
\[\q-\sum_{j=1}^l\alpha_j\q_j=\alpha_0\q_0,\]
where $\alpha_0$ is a constant. Since the answer of $\q_0$ is given, estimate with $\alpha_0\q_0$ always have better accuracy than with the combination of $\sum_{j=1}^h\beta_j\p_j$ (which are noisy answers). Therefore estimating $\q$ only relates to queries in $\W_i$ and $\q_0$ but not queries in $\W_1, \ldots, \W_{i-1},\W_{i+1},\ldots, \W_k$.
\end{proof}

\end{document}